\documentclass[12pt]{article}
\usepackage{amsthm,amsmath,amsfonts,amssymb}
\usepackage{graphicx}
\usepackage{enumerate}
\usepackage{natbib}
\usepackage{url} 

\newcommand{\blind}{0}

\newtheorem{theorem}{Theorem}[section]
\newtheorem{lemma}[theorem]{Lemma}
\newtheorem{remark}[theorem]{Remark}
\newtheorem{proposition}[theorem]{Proposition}
\newtheorem{alg}[theorem]{Algorithm}
\newtheorem{example}[theorem]{Example}

\begin{document}

\def\spacingset#1{\renewcommand{\baselinestretch}%
{#1}\small\normalsize} \spacingset{1}


\if0\blind
{
  \title{\bf Optimal and efficient circular designs\\
with neighboring effects}
  \author{Xiangshun Kong\\
    Department of Statistics,
Beijing Institute of Technology\\
    Xueru Zhang\\
    School of Statistics and Data Science \& LPMC,
Nankai University\\
    and \\
    Wei Zheng\thanks{
    The authors gratefully acknowledge \textit{please remember to list all relevant funding sources in the unblinded version}}\hspace{.2cm}\\
    Department of Business Analytics and Statistics,\\
University of Tennessee, Knoxville.}
  \maketitle
} \fi

\if1\blind
{
  \bigskip
  \bigskip
  \bigskip
  \begin{center}
    {\LARGE\bf Optimal and efficient circular designs\\
with neighboring effects}
\end{center}
  \medskip
} \fi

\bigskip
\begin{abstract}
This paper studies circular designs for interference models, where a treatment assigned to a plot also affects its neighboring plots within a block. For the purpose of estimating total effects, the circular neighbor balanced design was shown to be universally optimal among designs which do not allow treatments to be neighbors of themselves. Our study shows that these self-neighboring block sequences should be the main ingredient for an optimal design. Here, we adopt the approximate design framework and study optimal designs in the whole design space. Our approach is flexible enough to accommodate all possible design parameters, that is the block size and the number of blocks and treatments. It can be broken down into two main steps: the identification of the minimal supporting set of block sequences and the optimality condition built on it. The former is critical for reducing the computational time and memory usage tremendously. Unlike other related problems, there is no fixed pattern of the minimal supporting set here. To deal with this unprecedented challenge, we theoretically derived an intermediate set of sequences, which guarantees to contain the minimal supporting set. The latter can then be efficiently identified by a customized algorithm. Such combination of theoretical and algorithmic results is the first of its kind in related literature. Similar results are obtained for circular crossover designs. Lastly, our method is flexible enough to produce both symmetric and asymmetric designs and also to deal with 
arbitrarily forms of the within-block covariance matrix.
\end{abstract}

\noindent%
\vfill

\newpage
\spacingset{1.9} 
\section{Introduction}
\label{sec:intro}

In many applications of block designs, especially field experiments in agriculture, the treatment assigned to one plot may also have impacts on the responses of its neighboring plots. This phenomenon has long been recognized in literature, see \cite{stadler:1921,haines:1956,pearce:1957,cox:1958,bhalli:1964,jensen:1964,gomez:1972,dyke:1976,mcdonald:1976,cannell:1977,kawano:1978,murugesan:1978,kempton:1982,kempton:1984,jenkyn:1985,azais:1986,bradshaw:1986,kempton:1986,speckel:1987,bailey:1990,hide:1990,langton:1990,kempton:1992} for examples.  Among them, the treatment could be a plant variety, fertilizer, pesticide, virus type, or irrigation method. Typically, the blocks are arranged in rows of plots and hence the design of such experiments boils down to the determination of sequences of treatments for these blocks. It is most often assumed that the treatment applied to a plot has neighbor effects on its left and right immediate neighboring plots.

Besides neighbor effects, there are often {\it edge effects}, also called {\it border effects}, observed at the two ends of each block. \cite{langton:1990} commented that the edge effect might be caused by many unknown complex reasons and suggested that ``it will usually be essential to exclude from analysis all plot edges in order to ensure a fair comparison between treatments''. Correspondingly, \cite{azais:1993} introduced {\it circular designs}, in which two guarding plots at two ends of each block are set up to receive treatments without response so that the edge effects for the observed plots are totally under control. The name {\it circular} comes from the particular arrangement that the left guarding plot is assigned the same treatment as the right end plot and similarly for the opposite end, so that the two end plots appear like having neighbor effects on each other, hence the circular behavior. We refer to
\cite{aldred:2014,druilhet:1999,filipiak:2012a,filipiak:2005,zheng:2017} for detailed discussions on circular designs.

The majority of work on circular designs has focused on the estimation of direct effects, with neighbor effects being nuisance. On the other hand, it is desirable in practice to make the decision of selecting a single treatment to be applied over a larger spatial area. When the chosen treatment is in use, its only neighbor will be itself, and thus the parameter of interest shall be the sum of direct and neighbor effects. We shall call such effects as {\it total effects} as in \cite{bailey:2004}, which proved that the circular neighbor balanced design (CNBD) is universally optimal among the subclass of designs with no treatment as a neighbor of itself. Unfortunately, the constraint on the design space turns out to be severe: the efficiency of CNBD in the whole design space drops down to $50\%$ as the design size grows.  {Along this line of research, \cite{filipiak:2005} and \cite{ai:2009} further investigated the performance of CNBD when the within-block covariance is of AR(1) structure instead of the classical identify matrix. The only work that has lifted the constraint on the design space is \cite{druilhet:2012} who adopted the approximate design theory and found optimal designs among all designs, which of course allows self-neighboring block sequences. However, the results are still limited in the following sense. $(i)$ The derived optimal designs are all symmetric, and only exists when the number of blocks takes special values. Note CNBD is also a symmetric design. $(ii)$ The derivation of optimal designs boils down to solving a maximin problem over a set of representative block sequences, whose size grows superexponentially in the block length, say $k$. A general algorithm without further theoretical studies of the sequence structures can not deal with large values of $k$. Particularly, \cite{druilhet:2012} listed results of optimal designs when $k\leq 12$. With advancement of computational power nowadays, we can only push the limit up to $k=14$ in our own experience. $(iii)$ The within-block covariance is assumed to be proportional to the identity matrix. While the last one is relatively easy to be extended, it requires substantial understanding of the sequence structures and their impact on the parameter estimation in order to solve the first two issues, if possible.

In this paper, we provide a unified framework for deriving both symmetry and asymmetric designs for arbitrary covariance structures and arbitrary design sizes. Derive designs are optimal among all possible circular designs instead of restricting the comparison within a subclass, which confirms that self-neighboring sequences do play critical roles in optimal designs. Specifically, our results are comprehensive in following ways. }$(i)$ The approximate design theory is established for all possible combinations of $k$ and $t$, where $t$ is the number of treatments to be compared. Here, the most difficult part is to provide theoretical forms of the supporting sequences, especially when $k$ and $t$ are large. $(ii)$ We further allow the flexibility on the number of blocks, say $n$, and proposed methods to derive exact designs from the approximate design theory for an arbitrary $n$. In other words, all possible configurations of $(k,t,n)$ are covered without any combinatorial constraints. {On the contrary, the symmetric designs proposed by existing literatures require $n$ to be a multiple of $t(t-1)$.} It is also obvious that CNBD only exists when $t\geq k$. ($iii$) We allow the within-block covariance matrix to be any positive definite matrix while the past work mostly assumed this matrix to be proportional to the identity matrix. As a slight deviation from this identity assumption, \cite{filipiak:2005} and \cite{ai:2009} assumed AR(1) structure for the covariance matrix. They studied properties of CNBD for limited choices of $k$ and $t$. Meanwhile these results can not be generalized to other within-block covariance matrices. ($iv$) We provide answers of optimal designs for three different models. Models (\ref{model:3}) and (\ref{model:2}) both consider two-sided neighbors with the latter assuming the left and right neighbor effects being the same. Model (\ref{eqn:modelcrossover}) considers one-sided neighbor effects for crossover designs. Model (\ref{model:3}) is of the main interest here, but its intrinsic relationship with Model (\ref{model:2}) helps us derive the theoretical forms of the supporting sequences, which is the key for finding optimal designs. The results of optimal crossover designs for Model (\ref{eqn:modelcrossover}) are derived in the same way and hence will be briefly described.

The comprehensiveness of our results is achieved without the sacrifice of computational time. This is due to the combination of theoretical insights with efficient algorithms. Roughly speaking, linear equation systems regarding the proportions of all treatment sequences are established for universally optimal approximate designs. We further show that there is only a small subset of sequences, namely supporting sequences, allowed to have positive proportions and hence the computation of this linear equation system is tremendously reduced. In many cases, those supporting sequences can be theoretically identified following \cite{kushner:1997}'s arguments, where each sequence is associated with a quadratic function and the game became the identification of the minimax of these functions.  Unfortunately, for our problem such task becomes intractable. At the superficial level, there is no clear pattern of the supporting sequences that we can observe from computational results. This unusual phenomenon becomes the main hurdle for theoretical advancement. To tackle it, we first find a narrow enough interval which contains the minimax point instead of directly specifying its value. This partial result allows us to identify a slightly larger subset of sequences, where all supporting sequences must belong to. Then an algorithm is built to find the supporting sequences within this subset in $O(k^2)$ time. Without this subset, we would need $O(t^k)$ time to search for the supporting sequences.

Particularly, we find that optimal designs consist of sequences which allocate each treatment in a sub-block of adjacent plots with equal or almost equal numbers of replications. Unlike CNBD or designs derived in literature for direct effects, our proposed designs do not try to put as many treatments in a sequence as possible. The optimal number of distinct treatments in a sequence is around $\sqrt{2k}$ for crossover designs and $\sqrt{k}$ for interference models, whenever these numbers are smaller than the total number of treatments under consideration.

\if(0){
Relevant discussions have appeared in Patterson (1950,1951), McGilchrist (1965), McGilchrist and Trenbath (1971), Kempton (1985,1991,1997) and Besag and Kempton (1986) and Matthews (1988).}\fi

The rest of the paper is organized as follows. 
Section \ref{sec:formulation} formulates the design problems under the two interference models into a unified optimization problem. 
Section \ref{sec:general} theoretically establishes the approach to derive universally optimal designs.
In particular, Section \ref{sec:theory} provides two systems of linear equations to characterize all possible universally optimal approximate designs, one for symmetric designs and one for general designs. 
Section \ref{sec:exact} extends the results in Section \ref{sec:theory} to exact designs and provides a simplified algorithm to obtain the optimal designs.
The main results of this paper are in Section \ref{sec:sequence}, which derive theoretical forms of the supporting sequences to address the computational issues of these approaches especially for large designs. Results of different natures are separated into Sections \ref{sec:sequence2} and \ref{sec:largekt}.
Examples are provided in Section \ref{sec:example} to illustrate our theoretical results.
All proofs of theorems are given in supplementary materials.

\section{Problem formulation}\label{sec:formulation}
Throughout this paper, we consider designs on $\Omega_{k,t,n}$, the set of all block designs with $n$ blocks of size $k$ for the comparison of $t\geq 2$ treatments. We require $k\geq 4$ since no contrast of treatments is estimable for any circular design when $k\leq 3$. Suppose $Y_{dij}$ is the response observed from the $j$th plot of block $i$, we consider the following two models
\begin{eqnarray}
Y_{dij}&=&\mu+\beta_i+\tau_{d(i,j)}+\lambda_{d(i,j-1)}+\rho_{d(i,j+1)}+\varepsilon_{ij}.\label{model:3}\\
Y_{dij}&=&\mu+\beta_i+\tau_{d(i,j)}+\lambda_{d(i,j-1)}+\lambda_{d(i,j+1)}+\varepsilon_{ij},\label{model:2}
\end{eqnarray}
Here $\mu$ is the general mean, $\beta_i$ is the $i$th block effect, $d(i,j)$ is the treatment assigned to the $j$th plot of block $i$ by design $d$, $\tau_{d(i,j)}$ is the direct treatment effect of $d(i,j)$, $\lambda_{d(i,j-1)}$ is the neighbor effect of treatment $d(i,j-1)$ from the left neighbor, $\rho_{d(i,j+1)}$ denotes the neighbor effect from the right, and lastly $\varepsilon_{ij}$ is the error term with zero mean. Model (\ref{model:3}) reduces to Model (\ref{model:2}) if we assume $\lambda_i=\rho_i$, $1\leq i\leq t$, namely the neighbor effects are {\it undirectional}. For this reason, we call Models (\ref{model:3}) and (\ref{model:2}) as the directional and undirectional interference models, respectively. For both models, we consider circular designs, i.e., $\lambda_{d(i, 0)}=\lambda_{d(i, k)}$, $\lambda_{d(i, k+1)}=\lambda_{d(i, 1)}$ and $\rho_{d(i, k+1)}=\rho_{d(i, 1)}$.  Let $Y_d$ be the vector of responses organized block by block, these two models can be written in matrix forms of
\begin{eqnarray}
Y_d &=&1_{nk}\mu+U\beta+T_d\tau+L_d\lambda+R_d\rho+\varepsilon,\label{model:3'}\\
Y_d &=&1_{nk}\mu+U\beta+T_d\tau+L_d\lambda+R_d\lambda+\varepsilon,\label{model:2'}
\end{eqnarray}
where $\beta=(\beta_1,...,\beta_n)'$, $\tau=(\tau_1,...,\tau_t)'$, $\lambda=(\lambda_1,...,\lambda_t)'$, $\rho=(\rho_1,...,\rho_t)'$ with $'$ representing the transpose of a vector or a matrix. Also, $1_h$ represents a vector of $h$ ones, and $U=I_n\otimes 1_k$ with $\otimes$ being the Kronecker product and $I_k$ being the identity matrix of size $k$. Lastly, $T_d$, $L_d$ and $R_d$ represent the design matrices for the direct, left and right neighbor effects, respectively.

Our target here is to find the optimal design for the estimation of the total effect, namely $\phi=\tau+\lambda+\rho$ for Model (\ref{model:3}) and
$\phi=\tau+2\lambda$ for Model (\ref{model:2}). For this purpose, we shall re-parametrize those models as
\begin{eqnarray}
Y_{d}&=&1_{nk}\mu+U\beta+T_{d}\phi+\tilde{L}_{d}\lambda+\tilde{R}_{d}\rho+\varepsilon,\label{eq1-3}\\
Y_{d}&=&1_{nk}\mu+U\beta+T_{d}\phi+(\tilde{L}_{d}+\tilde{R}_{d})\lambda+\varepsilon,\label{eq1-2}
\end{eqnarray}
where $\tilde{L}_{d}=L_d-T_d$ and $\tilde{R}_{d}=R_d-T_d$. In other words, we have $\tilde{L}_{d}=(I_n\otimes H_{\rm c})T_d$ and $\tilde{R}_d=(I_n\otimes H_{\rm c}^{\prime})T_d$, where $H_{\rm c}=(\mathbb{I}_{i=j+1({\rm mod}~k)}-\mathbb{I}_{i=j})_{1\leq i,j\leq k}$ with $\mathbb{I}$ being the indicator function. Here we adopt a very mild condition for the covariance structure, i.e., $V(\varepsilon)=I_n\otimes \Sigma$ with $\Sigma$ being an arbitrary positive definite $k\times k$ matrix. By similar arguments as in \cite{kunert:1984}, the information matrix for $\phi$ under the two models are
\begin{eqnarray}\label{eqn:91503}
{~~~~~\rm Model~ (\ref{model:3}):} ~C_{d}&=&C_{d00}-
\left(
\begin{array}{cc}
C_{d01}  &   C_{d02}
\end{array}
\right)
\left(
\begin{array}{cc}
C_{d11}  &   C_{d12}\\
C_{d21}  &   C_{d22}
\end{array}
\right)^{-}
\left(
\begin{array}{c}
C_{d10}\\
C_{d20}
\end{array}
\right),\label{eqn:2173}\\
{~~~~~\rm Model~ (\ref{model:2}):} ~C_{d}&=&C_{d00}-(C_{d01}+C_{d02})\bigg(\sum_{i,j=1}^2C_{dij}\bigg)^{-}(C_{d10}+C_{d20}),\label{eqn:2172}
\end{eqnarray}
where $C_{dij}=G_i'(I_n\otimes \tilde{ B})G_j$, $0\leq i,j\leq 2$, with $G_0=T_d$, $G_1=\tilde{L}_d$, $G_2=\tilde{R}_d$, and $\tilde{B}=\Sigma^{-1}-\Sigma^{-1}J_k\Sigma^{-1}/1^{\prime}_k\Sigma^{-1}1_k$ with $J_k=1_k1_k'$. For technical conveniences, we shall define a projection matrix $B_k=I_k-k^{-1}J_k$. In fact, we have $\tilde{B}=B_k$ when $\Sigma=I_k$.

The block diagonal structure of the matrix $I_n\otimes \tilde{ B}$ allows us to write each $C_{dij}$ in an additive form, which induces the approximate design framework. Let us just examine $C_{d00}$ for illustration. With the block-wise decomposition $T_d=(T_1',T_2',...,T_n')'$, we have $C_{d00}=\sum^n_{i=1}T_i'\tilde{B}T_i$. Note that the summand $T_i'\tilde{B}T_i$ depends on block $i$ only through the sequence used in this block. Let ${\cal S}$ be the set of all $t^k$ treatment sequences, we shall denote $C_{s00}=T_i'\tilde{B}T_i$ if sequence $s\in {\cal S}$ is adopted in block $i$. By this notation, we have $C_{d00}=\sum_{s\in {\cal S}}n_sC_{s00}$, where $n_s$ is the number of times that sequence $s$ is selected in the design $d$. Similarly, we have $C_{dij}=\sum_{s\in {\cal S}}n_sC_{sij}$ for $0\leq i,j\leq 2$. This means that we can consider a design $d$ as a result of selecting $n$ sequences from ${\cal S}$ with replications, thus the representation $d=\{n_s : s\in {\cal S}\}$. Define its associated measure as $\xi_d=\{p_s : s\in {\cal S}\}$, where $p_s=n_s/n$ is the proportion of sequence $s$ in design $d$. Then we have $C_{dij}=nC_{\xi_dij}$ with $C_{\xi_dij}=\sum_{s\in {\cal S}}p_sC_{sij}$, and thus $C_d=nC_{\xi_d}$ with $C_{\xi_d}$ being derived from equations (\ref{eqn:2173}) and (\ref{eqn:2172}) by replacing $d$ therein by $\xi_d$. As a result, finding the optimal design $d$ boils down to finding the optimal measure $\xi_d$. In the approximate design framework, we shall relax $p_s$ to be any value in the interval $[0,1]$, in which case there does not necessarily exist an exact design $d$ associated with it. Thereafter, we shall suppress the subscript $d$ and aim to optimize $\xi$ over the measure space $\mathcal{P}=\{\{p_s:s\in\mathcal{S}\}: p_s\geq 0, \sum_{s\in {\cal S}} p_s=1\}$.

Following \cite{kiefer:1975}, we call a measure $\xi^*$ to be {\it universally optimal} if it maximizes $\Phi(C_{\xi})$ over $\mathcal{P}$ for any function $\Phi: \mathbb{R}^{t\times t}\rightarrow \mathbb{R}$ satisfying: ($C.1$) $\Phi$ is concave; ($C.2$) $\Phi(S'C_{\xi}S)=\Phi(C_{\xi})$ for any permutation matrix $S$; ($C.3$) $\Phi(bC_{\xi})$ is nondecreasing in the scalar $b>0$. Let ${\cal P}^*$ be the set of all universally optimal measures. Each element of ${\cal P}^*$ shall also be optimal under the alphabetical criteria of A, D, E, and T among others. If there exists an exact design $d$ with its associated measure $\xi_d\in {\cal P}^*$, then $d$ is said to be universally optimal. Otherwise, we shall produce exact designs based on the universally optimal measures and evaluate their performances by their efficiencies under the alphabetical criteria against a measure in ${\cal P}^*$.


\section{Optimal designs for general covariance matrix $\Sigma$}\label{sec:general}
\subsection{Approximate design theory}\label{sec:theory}



In view of the treatment exchangeability in condition ($C.2$) above, we shall call a collection of sequences to be an {\it equivalence class} if it is closed under any form of treatment relabeling of the sequences. To be specific, let ${\cal G}$ be the collection of all $t!$ possible permutations on the set $\mathbb{Z}_t=\{1,2,...,t\}$. In algebra, ${\cal G}$ is called the symmetric group of $\mathbb{Z}_t$. Then an equivalence class containing a representative sequence, say $s$, can be constructed by $\langle s \rangle=\{\sigma(s): \sigma\in {\cal G}\}$, where $\sigma(s)$ is a sequence generated by applying the treatment relabeling/permutation $\sigma$ on $s$. We call $\langle s \rangle$ as the equivalence class produced by sequence $s$. In fact, for any alternative sequence $\tilde{s}\in \langle s \rangle$, we have $\langle \tilde{s} \rangle=\langle s \rangle$ due to the group property of ${\cal G}$. As a result, two equivalence classes are either identical or mutually exclusive and we shall have the partition ${\cal S}=\cup^m_{i=1} \langle s_i \rangle$, where $s_i$'s are representative sequences for the $m$ distinct equivalence classes. To calculate $m$, note that one equivalence class could be represented by one way of partitioning $k$ balls into at most $t$ boxes. When $k\leq t$, $m$ is the well known Bell number depending only on $k$, that is, $m= k^kL(k)$ with $L(k)$ decaying at the exponential rate in $k$. Theorem \ref{prop:pd} examines the properties of these $m$ equivalence classes as related to the job of searching for optimal measures.

The main purpose of this section is to present Theorem \ref{thm:3}, which maps ${\cal P}^*$ to a linear subspace. To do that, we need to introduce some notations along with the definition of a special subset of measures, which overlaps with but does not contain ${\cal P}^*$. \cite{kushner:1997} called a measure to be symmetric if it assigns equal proportion to sequences within each equivalence class. One consequence here is that $C_{\xi ij}$ will be completely symmetric for all $0\leq i,j\leq 2$. So the name of symmetry could be justified by both the symmetric permutation ${\cal G}$ and these completely symmetric matrices. The latter is more relevant to the optimal design problem here since these matrices are the direct building blocks for computing the information matrix. Thus, we shall instead define a measure $\xi$ to be {\it symmetric} if all matrices $C_{\xi ij}, 0\leq i,j\leq 2$ are completely symmetric. Also, denote by $\mathcal{P}_0$ the collection of all symmetric measures. Alternative to the full permutation approach as in \cite{kushner:1997}, one can also construct a symmetric measure through an orthogonal array of type I. {See Example \ref{example4} in Section \ref{sec:example}.}

For any $\xi\in \mathcal{P}_0$, we have $C_{\xi ij}=c_{\xi ij}B_t/(t-1)$. 
This form can be found in view of the orthogonality between $B_t$ and $J_t$. Applying them to (\ref{eqn:2173}) and (\ref{eqn:2172}) yields
\begin{eqnarray}\label{eqn:218}
C_{\xi}=y_{\xi}B_t/(t-1), &~~~~ & y_{\xi}=c_{\xi 00}-\ell_{\xi}'Q_{\xi}^{-}\ell_{\xi},
\end{eqnarray} 
with $\ell_{\xi}$ and $Q_{\xi}$ defined as $(c_{\xi 01},c_{\xi 02})'$ and $(c_{\xi ij})_{1\leq i,j\leq 2}$ under Model (\ref{model:3}), reduced to $c_{\xi 01}+c_{\xi 02}$ and $\sum_{i,j=1}^2c_{\xi ij}$ under Model (\ref{model:2}). In this section, all theoretical results will apply to both Models (\ref{model:3}) and (\ref{model:2}) under these unified notations unless otherwise noted. (\ref{eqn:218}) indicates that a measure would be universally optimal within ${\cal P}_0$ if it maximizes the scaler $y_{\xi}$. In this regard, let $y^*=\max_{\xi\in {\cal P}} y_{\xi}$. Part $(iii)$ of Theorem \ref{prop:pd} takes it further and claims that such a measure is actually universally optimal among ${\cal P}$. In fact, it also says $C_{\xi}=y^*B_t/(t-1)$ for all $\xi\in {\cal P}^*$. Let $F_{\xi}=((c_{\xi 00},\ell_{\xi}')',(\ell_{\xi}',Q_{\xi}')')$ and define the support of a measure $\xi=\{p_s: s\in {\cal S}\}$ as ${\cal V}_{\xi}=\{s\in\mathcal{S}:p_s>0\}$. We have the following theorem. 
\if(0){\begin{table}[htp]
\begin{center}
{\bf we should not use the table to represent them anymore, and I agree with this}\\$F_{\xi}=
\left(
\begin{array}{cc}
c_{\xi 00}  &  \ell_{\xi}'    \\
\ell_{\xi}  & Q_{\xi}     
\end{array}
\right)
$, ~~~~~~~
\begin{tabular}{|c|c|c|}
\hline
 & Model (\ref{model:2}) & Model (\ref{model:3})\\\hline
$Q_{\xi}$ &$\sum_{i,j=1}^2c_{\xi ij}$&$(c_{\xi ij})_{1\leq i,j\leq 2}$\\\hline
$\ell_{\xi}$ &$c_{\xi 01}+c_{\xi 02}$&$(c_{\xi 01},c_{\xi 02})'$\\\hline
\end{tabular}~.
\end{center}
\label{default}
\end{table}%
}\fi
\begin{theorem}\label{prop:pd}
$(i)$ $Q_s=0$ if and only if $s\in \langle (11\ldots 11)\rangle$ and ${\rm rank}(Q_s)=1$ if and only if $s\in \langle(12\ldots12) \rangle$, which is only possible when $k$ is even. For any other sequence, we have $Q_s>0$. $(ii)$ For any $\xi\in {\cal P}^*$, we have $\langle (11\ldots 11)\rangle \cap {\cal V}_{\xi}=\emptyset$ and $Q_{\xi}>0$. $(iii)$  $\xi\in {\cal P}^*$ if and only if $C_{\xi}=y^*B_t/(t-1)$.
\end{theorem}

\begin{remark} The seminal work \cite{kushner:1997} is the first to use approximate design theory to study crossover designs. The argument critically relies on the condition that $Q_{s}$ is positive definite for all sequences. This condition no longer holds here. By Theorem \ref{prop:pd}, we get around this issue by first showing that any design $\xi$ with singular $Q_{\xi}$ can not be optimal. Such an idea is the first of its kind in the related literature. 
\end{remark}

Technically, the standard development of the linear equation system in Theorem \ref{thm:3} requires $Q_{\xi}>0$ for all $\xi\in {\cal P}$, i.e., $Q_s>0$ for all $s\in {\cal S}$, which is not the case here. We remedy this issue by singling out the only two equivalence classes of sequences without positive definite $Q_s$, i.e., $\langle (11\ldots 11)\rangle$ and $\langle(12\ldots12) \rangle$. We also claim that $\langle (11\ldots 11)\rangle$ can simply be ignored and $\langle(12\ldots12) \rangle$ need to be combined with other sequences to construct universally optimal measures. That is, the original measure space ${\cal P}$ can be shrunk a little bit by taking out these singularities, yet without missing out any element of ${\cal P}^*$. Particularly, we have ${\cal P}^*\subset {\cal P}^+$ with ${\cal P}^+=\{\xi\in {\cal P}: Q_{\xi}>0\}$.

Part $(iii)$ allows us to quickly check the universal optimality of a given measure, but extra tools are needed to identify the whole ${\cal P}^*$ in an efficient way. This is the task of the rest of this paper. Observe the linearity $c_{\xi ij}=\sum_{s\in {\cal S}}p_sc_{sij}$ with $c_{sij}=tr(B_tC_{s ij}B_t)$ for $0\leq i,j\leq 2$ due to the fact $C_{\xi ij}=\sum_{s\in {\cal S}}p_sC_{sij}$. Propagating this linearity forward, we have $\ell_{\xi}=\sum_{s\in {\cal S}}p_s\ell_{s}$, $Q_{\xi}=\sum_{s\in {\cal S}}p_sQ_{s}$, and $F_{\xi}=\sum_{s\in {\cal S}}p_sF_{s}$, with $\ell_{s}$, $Q_{s}$, $F_{s}$ equaling $\ell_{\xi}$, $Q_{\xi}$, $F_{\xi}$ when the measure ${\xi}$ is a degenerated measure with a single sequence $s$. Define the quadratic functions $q_s(x)=c_{s00}+2\ell_s'x+x'Q_sx$ and $q_{\xi}(x)=c_{\xi 00}+2\ell_{\xi}'x+x'Q_{\xi}x$ so that $q_{\xi}(x)=\sum_{s\in {\cal S}}p_s q_s(x)$, with $x\in\mathbb{R}^2$ for Model (\ref{model:3}) and $x\in\mathbb{R}$ for Model (\ref{model:2}). One can verify that $y_{\xi}=\min_{x}q_{\xi}(x)$ for $\xi \in {\cal P}$ and the minimum is achieved at $x_{\xi}=-Q_{\xi}^{-1}\ell_{\xi}$ for $\xi\in {\cal P}^+$. Define $r(x)=\max_{s\in{\cal S}} q_{s}(x)$, which is convex due to the convexity of $q_{s}(x)$. Hence it has an attainable minimum value denoted by $y_*=\min_{x}r(x)$. By Theorem \ref{prop:pd}, the minimizing point of $r(x)$ shall also be unique, and thus the notation:
\begin{eqnarray}\label{eqn:2262}
x^*&=&\arg\min_{x}r(x).
\end{eqnarray}
So obviously $r(x^*)=y_*$. The following results are useful for identifying $x^*$ and $y^*$.


\begin{theorem}\label{thm:1}
$(i)$ $y_*=y^*$. $(ii)$ $\xi\in {\cal P}_0\cap {\cal P}^*$ if and only if $det(F_{\xi})>0$ and
\begin{eqnarray}\label{eqn:226}
\max_{s\in {\cal S}}[tr(F_sF^{-1}_{\xi})-tr(Q_sQ^{-1}_{\xi})]&=&1.
\end{eqnarray}
Besides, the maximum can be achieved by all $s\in {\cal V}_{\xi}$. $(iii)$  $\xi\in {\cal P}_0\cap {\cal P}^*$ implies $-Q_{\xi}^{-1}\ell_{\xi}=x^*$ and $q_{\xi}(x^*)=y^*$.
\end{theorem}

Note that part ($ii$) takes the form of the well known general equivalence theory (GET). Equation (\ref{eqn:226}) permits the usage of many GET based algorithms, such as Federov's exchange algorithm, to derive a measure in ${\cal P}_0\cap {\cal P}^*$. Due to the invariance of the values of $c_{sij}$'s within each equivalence class, it is sufficient to only consider sequences in $\{s_1,s_2,...,s_m\}$ instead of ${\cal S}$ for both the maximization step and the measure updating step during the exchange algorithm. This treatment reduces the dimensionality of the optimizing problem from $|{\cal S}|=t^k$ to $m$. It allows us to find a measure in ${\cal P}_0\cap {\cal P}^*$ with the computational complexity of $O(m)$. Now $x^*$ and $y^*$ can be simply calculated based on part  ($iii$). Different measures in ${\cal P}_0\cap {\cal P}^*$ may have different $Q_{\xi}$ and $\ell_{\xi}$, but surely the same value of $Q_{\xi}^{-1}\ell_{\xi}$. Note there is no need to worry about the singularity issue of $Q_{\xi}$ in view of Theorem \ref{prop:pd}. Now we are ready to characterize all measures in ${\cal P}^*$ as follows.

\begin{theorem}\label{thm:3} 
Let ${\cal T}=\{s\in{\cal S}: q_{s}(x^*)=y_*\}$. Then $\xi\in {\cal P}^*$ if and only if
\begin{eqnarray}
\sum_{s\in {\cal T}}p_s[E_{s00}+E_{s01}(x^*\otimes B_t)]&=&y^*B_t/(t-1),\label{eq3-1}\\
\sum_{s\in {\cal T}}p_s[E_{s10}+E_{s11}(x^*\otimes B_t)]&=&0,\label{eq3-2}\\
\sum_{s\in {\cal T}}p_s&=&1,\label{eq3-3}
\end{eqnarray}
where the new notations $E_{s00},E_{s10}(=E_{s01}'),E_{s11}$ represent $C_{s00},(C_{s01},C_{s02}),(C_{sij})_{i,j=1,2}$ for Model (\ref{model:3}) and $C_{s00},C_{s01}+C_{s02},\sum_{i,j=1}^2C_{sij}$ for Model (\ref{model:2}).
\if(0){\begin{table}[htp]
\begin{center}
\begin{tabular}{|c|c|c|}
\hline
& Model (\ref{model:2}) & Model (\ref{model:3})\\\hline
$E_{s00}$ &$C_{s00}$&$C_{s00}$\\\hline
$E_{s01}$ &$C_{s01}+C_{s02}$&$(C_{s01},C_{s02})$\\\hline
$E_{s11}$ &$\sum_{i,j=1}^2C_{sij}$&$(C_{sij})_{i,j=1,2}$\\\hline
\end{tabular}
\end{center}
\label{default}
\end{table}%
}\fi
\end{theorem}

\subsection{Exact designs and algorithm}\label{sec:exact}
Equations (\ref{eq3-1})--(\ref{eq3-3}) indicate that it is sufficient to only consider sequences in ${\cal T}$ instead of ${\cal S}$ in the search of universally optimal measures. This reduces the computational burden tremendously. In general, ${\cal T}$ can be derived by the following steps along with the values of $x^*$ and $y^*$.
\begin{alg}\label{alg:T}{(For finding $x^*,y^*,{\cal T}$)} 
\begin{itemize}
\item[Step 0.] Specify representative sequences, $s_1$, ..., $s_m$, for each of the $m$ equivalence classes. 
\item[Step 1.] Maximize $y_{\xi}$ over all $\xi$ supported on $\{s_1,...,s_m\}$, and denote the maximizer by $\xi_0$.
\item[Step 2.] Calculate $x^*=-Q^{-1}_{\xi_0}\ell_{\xi_0}$ and $y^*=q_{\xi_0}(x^*)$.
\item[Step 3.] Identify the index set $A=\{i\in \mathcal{Z}_m: q_{s_i}(x^*)=y^*\}$.
\item[Step 4.] Recover ${\cal T}=\cup_{i\in A}\langle s_i\rangle$.
\end{itemize}
\end{alg}
Here, the maximization in step 1 can be achieved through an exchange algorithm based on the general equivalence theorem (GET) type of results in Theorem \ref{thm:1} with complexity $O(m)$. All other parts are calculated instantly. With the derived $(x^*,y^*,{\cal T})$, all universally optimal measures in ${\cal P^*}$ can now be recovered from the linear equations in Theorem \ref{thm:3}. In fact, these conditions can also be used to find an exact design that is either highly efficient or even optimal. To be specific, multiplying all terms in (\ref{eq3-1})--(\ref{eq3-3}) by $n$, we have 
\begin{eqnarray}\label{eqn:227}
\sum_{s\in {\cal T}}n_s
\left[
\begin{array}{c}
E_{s00}+E_{s01}(x^*\otimes B_t)-y^*B_t/(t-1)\\
E_{s10}+E_{s11}(x^*\otimes B_t)
\end{array}
\right]
&=&0
\end{eqnarray}
and $\sum_{s\in {\cal T}}n_s=n$ with all $n_s$ being non-negative integers. An exact design could be found by minimizing the Frobenius norm of the matrix on the left side of (\ref{eqn:227}) through an integer quadratic programming (IQP) solver such as {\it Gurobi}. As evidenced by simulation examples later in this paper, the exact designs such found are highly efficient or even optimal under various criteria. 

Alternatively, one can construct a symmetric exact design without resorting to IQP. The following theorem represents all symmetric universally optimal measures in terms of linear equations. This means (\ref{eqn-1})--(\ref{eqn-2}) is equivalent with (\ref{eqn:226}) when ${\cal S}$ therein is replaced by ${\cal T}$, but they serve different purposes. While the latter facilitates the calculation of $(x^*,y^*,{\cal T})$, the former helps the direct construction of symmetric measures or designs.

\begin{theorem}\label{thm:3.4}
For a symmetric measure $\xi=\{p_s: s\in {\cal S}\}\in {\cal P}_0$, we have $\xi \in {\cal P}^*$ if and only if
\begin{eqnarray}
\sum_{s\in {\cal T}} p_s (\ell_s+Q_sx^*)&=&0,\label{eqn-1}\\
\sum_{s\in {\cal T}} p_s&=&1\label{eqn-2}.
\end{eqnarray}
\end{theorem}

{Note there is only one or two linear equations to solve here and the calculation of the optimal proportions becomes trivial. Compared with the IQP approach, this approach is much faster, meanwhile it is not flexible in $n$, i.e., the number of blocks, since we have to make sure the design is symmetric. When the representative sequences with their associated weights are derived from (\ref{eqn-1})--(\ref{eqn-2}), we need to expand each representative sequence to a set of sequences to make sure the resulting design is symmetric. Typically, an orthogonal array of type I ($OA_I$) is used in this symmetrization step (see Example \ref{example4}). Due to the structure of $OA_I$, the number of blocks for such an exact design will be a multiple of $t(t-1)$. In fact, existing work on the current design problem has all been adopting this symmetric design approach.

Lastly, the methods laid out in this section is applicable to the crossover design with the model 
\begin{eqnarray}
Y_{dij}=\mu+\beta_i+\tau_{d(i,j)}+\lambda_{d(i,j-1)}+\varepsilon_{ij}.\label{eqn:modelcrossover}
\end{eqnarray}
Compared with Model (\ref{model:3}), we only have the neighbor effect from the left in (\ref{eqn:modelcrossover}). This is because the index $j$ represents the time and we only have the carryover effects from the past treatment rather than the future treatment. \cite{bailey:2004} gave the exact form of ${\cal T}$ for this model, but did not discuss the construction of exact designs. We shall illustrate our method of producing exact designs for this model in Examples \ref{example1} and \ref{example2}.


\section{Theoretical form of ${\cal T}$}\label{sec:sequence}
Based on the discussion in Section \ref{sec:exact}, it is crucial to derive the triplet $(x^*,y^*,{\cal T})$ for the search of optimal or highly efficient exact designs. Algorithm \ref{alg:T} is the state of art tool for finding $(x^*,y^*,{\cal T})$ in a general setup. However, this algorithm could still become infeasible as the design size further grows. Specifically, it has the complexity of $O(m)$, where $m$ is the number of different equivalence classes. Note $m$ increases superexponentially with respect to $k$ and $t$, and empirically we find the algorithm to be no longer affordable when $k\geq 15$. Hence there is a need to provide the theoretical form of ${\cal T}$ whenever possible. We achieve this by assuming $\Sigma=I$. This assumption is also adopted by existing literatures, wherein \cite{druilhet:2012} numerically tabulated optimal designs for $k\leq 12$ and listed the exact form of ${\cal T}$ for $k\leq 12$. Here we shall provide the theoretical form of ${\cal T}$ for all combinations of $k$ and $t$. Besides the computational benefit, such a theoretical form also provides insights into what forms of sequences are typically preferred. {The challenging part is that there is no pattern of supporting sequences from the computational result when $k\leq 14$ so that it is not easy to even guess about the forms supporting sequences through the aids of a computer.}

The discussion is carried out in two parts. Section \ref{sec:sequence2} deals with the cases when the set ${\cal T}$ can be directly described. This is possible when $t=2,3$ with any $k$. In all these cases, the cardinality of $\cal T$ is reasonably small. The much more complicated situation of $t>3$ is studied separately in Section \ref{sec:largekt}, where two types of results are presented. In Section \ref{sec:theoreticalT}, a subset which contains ${\cal T}$ but much smaller than ${\cal S}$ is provided. As a result, an algorithm can further be used to recover ${\cal T}$ from this subset very quickly. Alternatively in Section \ref{sec:higheff}, we provide a substitute of ${\cal T}$ which contains only one equivalence class but produces highly efficient designs.

\subsection{The straightforward cases: $t=2, 3$}\label{sec:sequence2}
As mentioned above, the triplet $(x^*,y^*,{\cal T})$ plays an essential role in finding optimal measures or designs. Even though we have laid out a general idea to derive them without knowing the optimal measure, it will still become a daunting work to carry out the computation when $k$ and $t$ are large. Hence it is crucial to know their theoretical forms in such cases. These results are also important for obtaining insights into preferred structures of block sequences for any design size. To do that, we first need to distinguish the notation $(x^*,y^*,{\cal T})$ for Models (\ref{model:3}) and (\ref{model:2}) by $(x^*_1,y^*_1,{\cal T}_1)$ and $(x^*_2,y^*_2,{\cal T}_2)$, respectively. In this section, Theorem \ref{thm:5} establishes the connection between these two triplets as well as the connection of the corresponding optimal measures. Such connections guide us to study these two models in an interactive way. 
In Theorem \ref{thm:table}, for $t\leq 3$, we derive the explicit expression of $(x^*,y^*,{\cal T})$ which shows that the cardinality of $\cal T$ is reasonably small. Section \ref{sec:largekt} studies the case of $k>10$ and $t>3$ that is not covered by Theorem \ref{thm:table}.

To introduce Theorem \ref{thm:5}, we define the dual of a sequence $s=(t_1,t_2,...,t_k)$ by reversing the positions of the treatments, that is $s'=(t_k,...,t_1)$. With the definition $p_{\langle s\rangle}=\sum_{\tilde{s}\in \langle s\rangle}p_{\tilde{s}}$, a measure is said to be \emph{self-dual} if $p_{\langle s\rangle}=p_{\langle s^{\prime}\rangle}$ for any $s\in{\cal S}$.


\begin{theorem}\label{thm:5}
If $\Sigma$ is persymmetric, the following hold.
\begin{itemize}
\item[($i$)] $x^*_1=(x^*_2, x^*_2)$, $y^*_1=y^*_2$ and ${\cal T}_1={\cal T}_2$.
\item[($ii$)] For any criterion function satisfying $(C.1)$--$(C.3)$, the efficiency of any given measure under Model (\ref{model:2}) is greater than or equal to that under Model (\ref{model:3}).
\item[($iii$)] The efficiency of a symmetric self-dual measure is same under Models (\ref{model:3}) and (\ref{model:2}).
\end{itemize}
\end{theorem}

Part ($i$) indicates that we can find ${\cal T}_1$ by working on the simpler task of finding ${\cal T}_2$. Part ($ii$) indicates that it is sufficient to work solely on Model (\ref{model:3}), and the derived measure is automatically suited for Model (\ref{model:2}). Particularly, the universal optimality of a measure under Model (\ref{model:3}) implies its universal optimality under Model (\ref{model:2}), and its reverse is implied by part ($iii$). The total effects are not estimable for any $\Sigma$ under Models (\ref{model:3}) and (\ref{model:2}) when $k\leq3$. These results not only help get around the computational bottleneck for large $k$ and $t$, but also provides insight on preferred arrangements of treatments within a block. To precede, we call a matrix to be of type-H if it can be expressed in the form $a I_{k}+b 1_{k}^{\prime}+1_{k} b^{\prime}$ for $a \in$ $\mathbb{R}^{+}$ and $b \in \mathbb{R}^{k}$. It covers the special case of completely symmetric matrices and the most often adopted case of identity matrix in relevant literature.

\begin{theorem}\label{thm:table}
Suppose $\Sigma$ is of type-H in Model (\ref{model:2}), ${\cal T}_2$ and $x_2^*$ are derived and displayed in Table \ref{tb:theorem} for cases of ($i$) $t=2$ and $k\geq 4$, ($ii$) $t=3$ and $k\geq 4$.
If $\Sigma$ is also persymmetric, we have: $x^*_1=(x^*_2, x^*_2)$ and ${\cal T}_1={\cal T}_2$ for Model (\ref{model:3}).
\end{theorem}
Note that the results in Table \ref{thm:table} still work for Model (\ref{model:2}) if we remove the dual sequences $s_a'$ and $s_b'$ since $c_{sij}=c_{s'ij}$ for $0\leq i,j\leq 2$ under this model. 
\begin{center}
\begin{table}
\caption{Representative sequences of ${\cal T}_1$ and $x^*_1$ for Theorem \ref{thm:table} with integers $\lambda\geq 2$, $\mu\geq 3$.}\label{tb:theorem}
\begin{tabular}{|c|c|c|c|c|c|}
  \hline
  $t$&$k$ & representative sequences of equivalence classes in ${\cal T}_1$ & $x_1^*$\\ \hline
$t=2$  &$2\lambda $& $s_{a}=(1_{\lambda}^{\prime}|2\cdot1_{\lambda}^{\prime})$, $s_b=(1_{\lambda}^{\prime}\otimes(12))$  & $(\lambda-1)/(2\lambda-1)$ \\
  &$2\lambda+1 $&  $s_a=(1_{\lambda+1}^{\prime} |2\cdot1_{\lambda}^{\prime})$, $s_b=(1_{\lambda}^{\prime}\otimes(12)|1)$ &$(\lambda-2)/(2\lambda-3)$\\\hline
$t= 3$&4  &  $s_a=(1123)$, $s_b=(1213)$ &$1/3$ \\
  &5 & $s_a=(11223)$, $s_b=(11232)$, $s_c=(12323)$ & $2/5$ \\
  &$6$ & $s_a=(112233)$, $s_b=(112323)$ & $2/5$\\
  &7 & $s_a=(1112223)$, $s_b=(1112323)$ & $(28+\sqrt{532})/126$ \\

      &$8$ & $s_a=(11122333)$, $s_b=(11112323)$, $s_b'$ & $3/7$\\
  &$3\mu$ & $s_a=(1^{\prime}_{\mu}|2\cdot1^{\prime}_{\mu}|3\cdot1^{\prime}_{\mu})$, $s_b=(1^{\prime}_{\mu+1}|1^{\prime}_{\mu-1}\otimes(23)|2)$, $s_b'$ & $(2\mu-2)/(4\mu-3)$\\
  &$3\mu+1$ & $s_a=(1^{\prime}_{\mu}|2\cdot1^{\prime}_{\mu+1}|3\cdot1^{\prime}_{\mu})$, $s_b=(1^{\prime}_{\mu+1}|1^{\prime}_{\mu}\otimes(23))$,  $s_b'$& $(2\mu-2)/(4\mu-3)$\\
  &$3\mu+2$& $s_a=(1^{\prime}_{\mu+1}|2\cdot1^{\prime}_{\mu}|3\cdot1^{\prime}_{\mu+1})$, $s_b=(1^{\prime}_{\mu+1}|1^{\prime}_{\mu}\otimes(23)|2)$, $s_b'$ & $(2\mu-1)/(4\mu-1)$ \\
  \hline
\end{tabular}
\end{table}
\end{center}

\subsection{The complicated case: $t>3$}\label{sec:largekt}
The purpose of this section is to study $\cal T$ with $t> 3$, that is not yet covered in the previous section. The reason for this separated investigation is because there is no clear pattern of $\cal T$ as in the earlier cases. As a result, this large design scenario can not be determined as directly as in Table \ref{tb:theorem}. This explains why there is a lack of general theory for a broad range of $k$ and $t$ in literature. To the best of our knowledge, the exact form of ${\cal T}$ has not yet been given in any literature for circular designs in estimating total effects when $k$ and $t$ are large. We provide two types of results here. In Section \ref{sec:theoreticalT}, a subset which contains ${\cal T}$ but much smaller than ${\cal S}$ is provided. As a result, an algorithm can further be used to recover ${\cal T}$ from this subset very quickly. Alternatively in Section \ref{sec:higheff}, we find efficient measures based on a single equivalence class. Even with this simplified approach, the pattern of the efficient measures only begins to reveal itself when $k$ is beyond $20$. As will be shown in Theorem \ref{thm:11}, the efficiency of our proposed measure converges to $1$ at the rate of $0.04/\sqrt{0.96k}$.

\subsubsection{The exact form of ${\cal T}$}\label{sec:theoreticalT}
In this section, we find a small subset of sequences which contains ${\cal T}$ so that ${\cal T}$ can be recovered from this subset with the computational complexity of $O(k^2)$. First, we define a sequence $s(k,k_1,t_1,t_2)$ of length $k$ in a recursive way as follows. 
\begin{eqnarray}
&&\quad\tilde{s}(k,t_1,t_2)=(~(t_1+1)\cdot1_{f_{t_1+1}}'~|~(t_1+2)\cdot1_{f_{t_1+2}}'~|~\cdots~|~(t_1+t_2)\cdot1_{f_{t_1+t_2}}'), \label{eqn:121001}\\
&&\quad\quad {\rm where}~ 1+f_{t_1+t_2}\geq f_{t_1+1}\geq f_{t_1+2}\cdots\geq f_{t_1+t_2-1}\geq f_{t_1+t_2}, ~{\rm and}~ \sum^{t_1+t_2}_{j={t_1+1}}f_j=k;\notag\\
&&\quad\hat{s}(k,t)={\cal M}(~\tilde{s}(k-\lfloor k/2\rfloor,0,t-\lfloor t/2\rfloor),~\tilde{s}(\lfloor k/2\rfloor,t-\lfloor t/2\rfloor,\lfloor t/2\rfloor)~);\label{eqn:121002}\\
&&\quad s(k,k_1,t_1,t_2)=(~\hat{s}(k_1,t_1)~|~\tilde{s}(k-k_1,t_1,t_2)~). ~~(0\leq k_1\leq k) \label{eqn:121003}
\end{eqnarray}
The sequence in (\ref{eqn:121003}) is uniquely determined by parameters $(k,k_1,t_1,t_2)$, and will be called a {\it candidate sequence} hereafter. In (\ref{eqn:121002}), we used an operator ${\cal M}$ which intertwines two sequences. Specifically, for sequences ${a}=(a_1,\ldots,a_p)$ and ${b}=(b_1,\ldots,b_p)$ of the same length $p$, we have ${\cal M}({a},{b})=(a_1,b_1,\ldots,a_p,b_p)$. If instead ${b}=(b_1,\ldots,b_{p-1})$ is of length $p-1$, we have ${\cal M}({a},{b})=(a_1,b_1,\ldots,a_{p-1},b_{p-1},a_p)$. {For a candidate sequence to be well defined, we require 
sequences (\ref{eqn:121001}) and (\ref{eqn:121002}) vanish to {\it empty} sequences with zero length if $k=0$, and $t_2>0$ when $k-k_1>0$ and $k_1,t_1\geq 2$ when $k_1>0$ in (\ref{eqn:121003}).} We give a toy example to illustrate the construction of the candidate sequence. Suppose $(k,k_1,t_1,t_2)=(21,13,4,3)$. Then $\hat{s}(k_1,t_1)=\hat{s}(13,4)={\cal M}(\tilde{s}(7,0,2),\tilde{s}(6,2,2))={\cal M}((1,1,1,1,2,2,2),$ $(3,3,3,4,4,4))=(1,3,1,3,1,3,1,4,2,4,2,4,2)$ and we have $\tilde{s}(k-k_1,t_1,t_2)=\tilde{s}(8,4,3)=(5,5,5,6,6,6,7,7)$. Finally, we generate $s(21,13,4,3)=(1,3,1,3,1,3,1,4,2,4,2,4,2,5,$ $5,5,6,6,6,7,7)$. Hence $k=21$ is the length of the sequence $s(k,k_1,t_1,t_2)$ which contains $t_1+t_2=7$ treatments, $k_1=13$ is the length of the intertwined subsequence $\hat{s}(k_1,t_1)$ which contains $t_1=4$ treatments. First, we show that the values of $x^*$ and $y^*$ can be computed based on a quite small subset of sequences as defined by 
$${\cal S}^*=\{s(k,k_1,t_1,t_2): 0\leq t_1\leq t_1+t_2\leq \min(4\sqrt{k}+2,t),0\leq k_1\leq k\}.$$
Meanwhile, as will be shown in Theorem \ref{thm:121101}, ${\cal S}^*$ also plays a crucial role for recovering ${\cal T}$.

\begin{theorem}\label{lem:121001}
Suppose $\Sigma$ is of type-H. For $k>10$ and $t>3$, we have
\begin{eqnarray}
y^*&=&\min_x\max_{s\in{\cal S}^*}q_s(x);\label{eqn:121005}\\
x^*&=&\arg\min_x\max_{s\in{\cal S}^*}q_s(x).\label{eqn:121006}
\end{eqnarray}
\end{theorem}

In fact, we can show that $|{\cal S}^*|\leq 10k^2$ for all $k> 12$ and any $t$. This cardinality is much smaller than the original full set $|{\cal S}|=t^k$. It is now very time efficient to apply Algorithm \ref{alg:T} to obtain $x^*$ and $y^*$ by solving (\ref{eqn:121005}) and (\ref{eqn:121006}). Consider $k=100$ and $t=5$, we have $|S|=7.88\times 10^{69}$ and it is impossible to directly execute Algorithm \ref{alg:T}. Instead, we have $|{\cal S}^*|\leq 10^5$, and we  instantaneously $x^*$ and $y^*$ by working on ${\cal S}^*$. Next, we will discuss how to recover ${\cal T}$ based on $x^*,y^*,{\cal S}^*$.
To proceed, we need some new notations. For $s=(t_1,\ldots,t_k)$, recursively define
\begin{eqnarray}
\psi_{t'}(s)&=&|\{j:t_{j-1}=t'=t_{j+1},1\leq j\leq k\}|,~~\psi(s)=\sum_{t'=1}^t\psi_{t'}(s);\label{eqn:121301}\\
\gamma_{t'}(s)&=&|\{j:t_j=t'=t_{j-1},1\leq j\leq k\}|,~~\gamma(s)=\sum_{t'=1}^t\gamma_{t'}(s);\label{eqn:121302}\\
\chi(s)&=&\sum_{t'=1}^tf_{s,t'}^2,~~f_{s,t'}=|\{j:t_{j}=t',1\leq j\leq k\}|;\label{eqn:121303}\\
{\cal C}(s)&=&\{s':\psi(s')=\psi(s),\gamma(s')=\gamma(s),\chi(s')=\chi(s)\}.\label{eqn:121401}
\end{eqnarray}
The notation $\psi_{t'}(s)$ and $\gamma_{t'}(s)$ are not redundant and will be used for proofs in supplementary materials. Note that, in (\ref{eqn:121301}) and (\ref{eqn:121302}), $t_0=t_k$ and $t_{k+1}=t_1$ under circular setup. Simple analysis reveals that $\langle s\rangle\subset{\cal C}(s)$.
Thus, ${\cal C}(s)$ is called the {\it pseudo equivalence class} of $s$. 
Let ${\cal C}^*=\{s\in{\cal S}^*:q_{s}(x^*)=\max_{s'\in{\cal S}^*} q_{s'}(x^*)\}$.

\begin{theorem}\label{thm:121101}
Suppose $\Sigma$ is of type-H. For $k>10$ and $t>3$, we have
\begin{eqnarray}\label{eqn:121107}
{\cal T}=\cup_{s\in{\cal C}^*}{\cal C}(s).
\end{eqnarray}
\end{theorem}

The idea of constructing the intermediate set ${\cal S}^*$ followed up with the algorithmic derivation of ${\cal T}$ as summarized in Theorem \ref{thm:121101} is the first of its kind in dealing with the issue when the points $(x^*,y^*)$ can not be exactly identified. The complexity of $\cal T$ depends on the cardinality of ${\cal C}^*$, which is normally very small. We have tried $k\leq 10^3$ with $t\leq k$ and found that ${\cal C}^*$ contains at most $5$ different sequences from ${\cal S}^*$. Now let us revisit the case of $(k,t)=(100,5)$. It only takes $3.71$ seconds on an ordinary 2.9 GHz MacBook Pro to identify ${\cal T}$ based on (\ref{eqn:121107}). Specifically, we have ${\cal C}^*=\{s_1,s_2\}$ with $s_1=s(100,0,0,5)=(1_{20}',2\cdot 1_{20}',\ldots,5\cdot 1_{20}')$, and $s_2=s(100,37,2,3)=(1,2,1,\ldots,1,2,1~|~3\cdot 1_{21}'~|~4\cdot 1_{21}'~|~5\cdot 1_{21}')$ where the length of the subsequence $(1,2,1,\ldots,1,2,1)$ is $37$. Meanwhile, we notice ${\cal C}(s_1)=\langle s_1\rangle$, ${\cal C}(s_2)=\langle s_2\rangle$. Hence we simply have ${\cal T}=\langle s_1\rangle\cup\langle s_2\rangle$. In fact, the computational cost does not grow with $t$. Now take $t=20$ instead of $5$, we have ${\cal C}^*=\{s_3,s_4\}$ where $s_3=s(100,0,0,10)=(1_{10}',2\cdot 1_{10}',\ldots,10\cdot 1_{10}')$ and $s_4=s(100,70,10,3)=({\cal M}(1_{7}',6\cdot 1_{7}')|\cdots|{\cal M}(5\cdot 1_{7}',10\cdot 1_{7}')|11\cdot 1_{10}'|12\cdot 1_{10}'|13\cdot 1_{10}')$.

\subsubsection{Efficient designs based on a single equivalence class}\label{sec:higheff}

To build efficient designs based on a relatively simpler supporting set, say ${\cal R}$, than ${\cal T}$. The quality of ${\cal R}$ can be evaluated by its efficiency $e_{\cal R}=y_{\cal R}/y^*$, where $y_{\cal R}=\max_{\xi:{\cal V}_{\xi}\subset \cal R}y_{\xi}$. In other words, $e_{\cal R}$ calculates the efficiency of the best measure we could potentially construct using sequences in ${\cal R}$. To prepare for following theorems, we need some notations. For $1\leq i\leq \min(k,t)$, let
\begin{eqnarray}\label{eqn:si}
s_i=(1_{f_1}',2\cdot 1_{f_2}',\cdots,i\cdot1_{f_i}') 
\end{eqnarray}
with the constraint of $1+f_i\geq f_1\geq f_2\cdots\geq f_{i-1}\geq f_i$ and $\sum^i_{j=1}f_j=k$. When $i$ divides $k$, treatments $1,2,...,i$ have equal replications. Otherwise, we assign one more replication to treatments in the left end of the sequence in order. The left end arrangement is only for ease of presentation. Due to the circular nature, we can cyclically shift the treatments in the sequence $s_i$ toward either direction at an arbitrary distance and the resulting sequence has the same statistical property as we will show later. Define the integers $i_0=\min\{{\arg\max}_{i} q_{s_i}(0.4),t\}$ and $i^*=\min\{{\arg\max}_{i} q_{s_i}(0.5),t\}$. The constants $0.96$ and $0.04$ in the following theorem are well engineered to work for all $k$ and $t$ as specified in the theorem. Efficient single sequences are first empirically found in the numerical results in \cite{druilhet:2012} with $k\leq 12$. Our result covers all possible $(k,t)$ and may coincide with those given by \cite{druilhet:2012} in few special cases.


\begin{theorem}\label{thm:11}
Suppose $\Sigma$ is of type-H. For $k>10$, let ${\cal R}{(k,t)}=\langle s_{i^*}\rangle$ under Model (\ref{model:2}) and ${\cal R}{(k,t)}=\langle s_{i^*}\rangle\cup\langle s'_{i^*}\rangle$ under Model (\ref{model:3}). We have $e_{{\cal R}{(k,t)}}> 1-v(t,k)$, where 
\begin{eqnarray}\label{eqn:309}
v(t,k)=\frac{0.04i_0}{k-k/i_0-0.96i_0-0.25i_0/k}
\end{eqnarray}
decreases in $k$ for any given $t$. This bound is asymptotically tight in the sense that $(1-e_{{\cal R}(k,t)})\sqrt{k}\rightarrow 0.04\sqrt{0.96}$ and $v(k,t)\sqrt{k}\rightarrow 0.04\sqrt{0.96}$ as $k\rightarrow\infty$ for any given $t$. {For $4\leq k\leq 10$, let ${\cal R}{(k,t)}=\langle s_a\rangle$ where $s_a$ is defined in Table \ref{tb:theorem}. We list the corresponding $e_{{\cal R}{(k,t)}}$ as follows.}
\begin{table}[htp]
\begin{center}
\caption{Efficiency of ${{\cal R}{(k,t)}}$ in Theorem \ref{thm:11} when $4\leq k\leq 10$.}
\setlength{\tabcolsep}{4mm}
\begin{tabular}{|c|c|c|c|c|c|c|c|}
\hline
  &$k=4$ &$k=5$ &$k=6$ &$k=7$&$k=8$ &$k=9$ &$k=10$ \\ \hline
 $t=2$ &-- &$0.8333$ &$0.9259$ &$0.9830$&$0.9800$ &$0.9952$ &$0.9918$ \\ \hline
 $t=3$ &$0.9000$ &$0.9821$ &$0.8929$ &$0.9956$&$0.9735$ &$0.9878$ &$0.9898$ \\ \hline 
 $t=4$ &$1.0000$ &$0.9098$ &$1.0000$ &$0.9956$&$0.9615$ &$0.9818$ &$0.9795$ \\ \hline
 $t\geq 5$ &$1.0000$ &$0.9097$ &$0.8819$ &$0.9956$&$0.9615$ &$0.9797$ &$0.9795$ \\ \hline
\end{tabular}
\end{center}
\label{default}
\end{table}%

\end{theorem} 

\begin{figure}
\center
\includegraphics[width=6in]{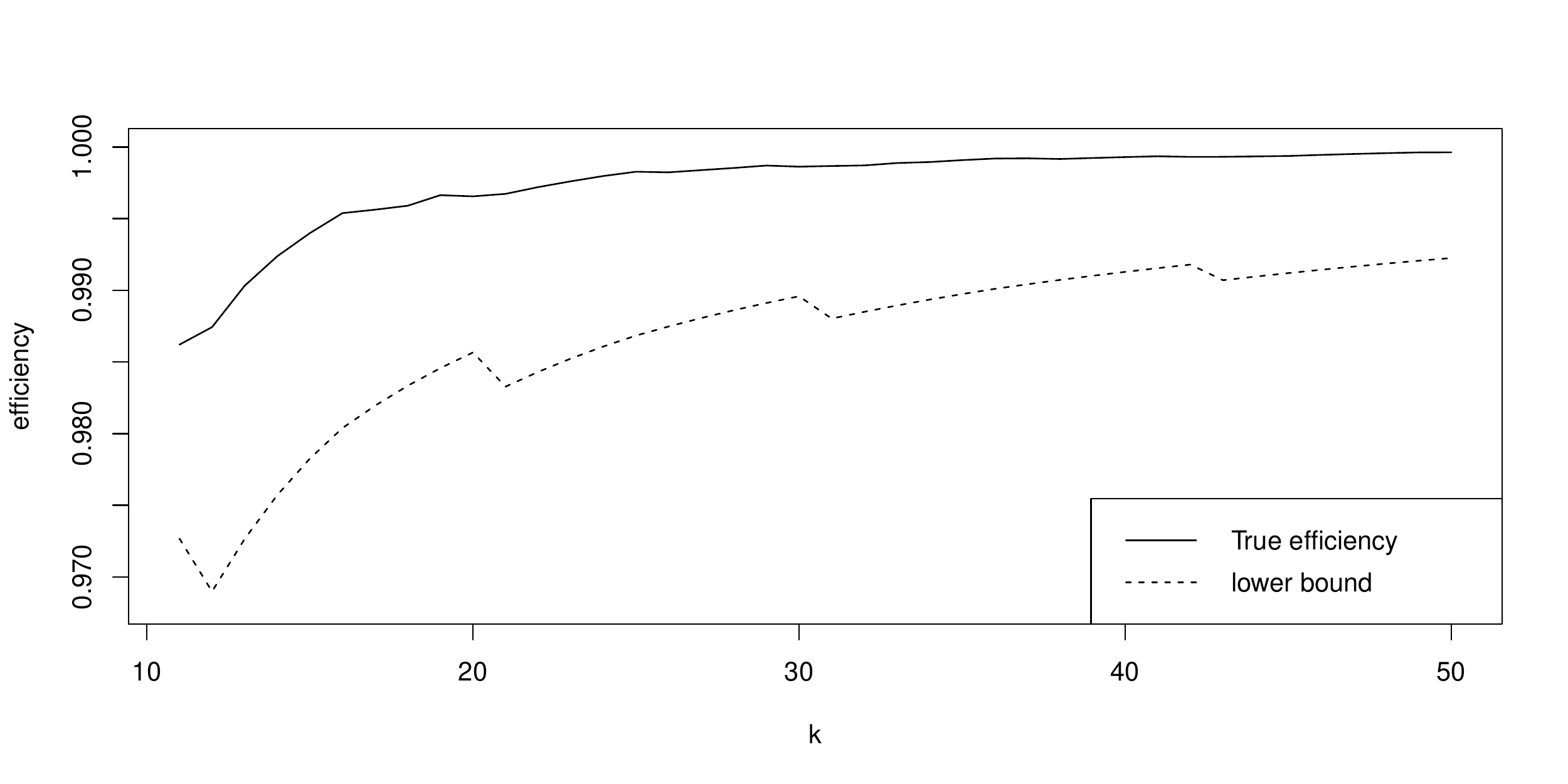}
\caption{The lower bound of $e_{{\cal R}(k,t)}$ derived in Theorem \ref{thm:11} and the true efficiency derived by Theorem \ref{lem:121001} in Section \ref{sec:largekt} for $t=7$ and a spectrum for $k$, i.e., $11\leq k\leq 50$.}
\label{fig:er}
\end{figure}

Theorem \ref{thm:11} does not cover the case $(k,t)=(4,2)$ since any single equivalence class alone leads to zero information matrix in this case.
Theorem \ref{thm:11} derives a lower bound of $e_{{\cal R}(k,t)}$ for our proposed supporting set ${\cal R}(k,t)$ without knowing the true ${\cal T}$ or optimal designs. It further shows that this bound is asymptotically tight. In Section \ref{sec:largekt}, $\cal T$ can be derived for arbitrary $(k,t)$, which enables us to generate true optimal designs and give the true efficiency of ${\cal R}(k,t)$. Figure \ref{fig:er} displays true efficiencies and efficiency bounds given by (\ref{eqn:309}) for the particular case of $t=7$ with $11\leq k\leq 50$, keeping in mind that the change of the value of $t$ does not affect the comparison of these two curves very much. By inspecting all values of $k,t\leq 10^4$, we observe that the high-efficiency equivalence class proposed in Theorem \ref{thm:11} is always a subset of $\cal T$. Moreover, by solving (\ref{eqn:227}), this equivalence class is always assigned a large weight in the true optimal design.

Note that CNBD is optimal among designs where a treatment is not a neighbor of itself at distances $1$ and $2$. Theorem \ref{thm:11} strongly indicates the necessity of including self-neighboring sequences in a design, if they are not the main and only important sequence. Such observation is validated by Theorems \ref{thm:table} and \ref{thm:11} for circular designs. This could also be verified by the results in \cite{druilhet:2012} when $k\leq 12$. 

\if{0}
\begin{theorem}\label{thm:120401}
Suppose $\Sigma$ is of type-H. For $k>10$ and $t>3$, optimal designs for Models (\ref{model:2}) and (\ref{model:3}) can be found on ${\cal R}_2(k,t)$, i.e., $e_{{\cal R}_2(k,t)}=1$ where
\begin{eqnarray}
{\cal R}_2(k,t)=&&\{(~\hat{s}(1,f_1),\ldots,\hat{s}(t_1,f_{t_1})~|~\tilde{s}(t_1+1,f_{t_1+1},f_{t_1+2}),\\
&&\quad\quad \tilde{s}(t_1+3,f_{t_1+3},f_{t_1+4}),\ldots,\tilde{s}(t_1+2r-1,f_{t_1+2r-1},f_{t_1+2r})~):\notag\\
&&1+f_{t_1}\geq f_1\geq f_2\cdots\geq f_{t_1-1}\geq f_{t_1}; \label{eqn:120702}\\
&&1+f_{t_1+2p}\geq f_{t_1+1}\geq f_{t_1+2}\cdots\geq f_{t_1+2p-1}\geq f_{t_1+2p}~~\geq 1; \notag\\
&&~\sum^{t_1}_{j=1}f_j=k_1;~\sum_{j=1}^{2r}f_{t_1+j}=k_2;~k_1+k_2=k,~t_1+2r\leq \min(3k^{1/2},t)\}.\label{eqn:120501}
\end{eqnarray}
\end{theorem}

Theorem \ref{thm:120401} helps researchers to focus on only two equivalence classes in finding optimal designs, one of the pattern (\ref{eqn:si}) and the other belongs to the ${\cal A}(k,t)$ defined in (\ref{eqn:120702})--(\ref{eqn:120501}). The following Algorithm \ref{alg:newalg1} further addresses the two equivalence classes in Theorem \ref{thm:120401} with computational burden $O(k^{2.5})$. This is usually affordable and negligible compared with $|\cal S|$.

From the definition, any sequence in ${\cal A}(k,t)$ contains two parts: $s_1(t_1,k_1)$ of pattern (\ref{eqn:si}), and $s_2(t_1,p,k_2)$ containing $2p$ treatments with nearly balanced frequencies. Note that $s_1(t_1,k_1)$ is uniquely determined by $t_1$ and $k_1$, and $s_2(t_1,p,k_2)$ is uniquely determined by $t_1$, $p$ and $k_2=k-k_1$. Thus, any sequence in ${\cal A}(k,t)$ can be uniquely determined by $(t_1,p,k_1)$ since $k$ is pre-specified in applications. However, it should be emphasized here that many combinations of $(t_1,p,k_1)$ will not appear in ${\cal A}(k,t)$ since (\ref{eqn:120702})--(\ref{eqn:120501}) obviously put some restrictions on $(t_1,p,k_1)$. For convenience, in the pseudo code of Algorithm \ref{alg:newalg1}, we define $\Psi:R^3\rightarrow {\cal S}$ as
$\Psi(t_1,p,k_1)=(~s_1(t_1,k_1)~|~{s_2(t_1,2p,k-k_1)}~)$ if there exists an $(~s_1(t_1,k_1)~|~{s_2(t_1,2p,k-k_1)}~)\in{\cal A}(k,t)$ for given $k$.
Otherwise, we intentionally set $\Psi(t_1,p,k_1)=\mathtt{NULL}$.
Some complex details and explanation of the searching range of parameters in Algorithm \ref{alg:newalg1} are postponed to Appendix for a better reading experience.

\begin{algorithm}
\caption{\scriptsize{(pseudo code to uniquely determine the two equivalence classes in Theorem \ref{thm:120401})}}\label{alg:newalg1}
\label{alg:A}
\begin{algorithmic}
\STATE {Initialize $i=1$, $k_1=0$, $t_1=1$, $p=1$. Set $\mathtt{Record.matrix}=\mathtt{NULL}$.} 
\FOR{each $2\leq i\leq \min(3\sqrt{k},t)$}
\STATE generate $s_i$ and calculate the explicit form of $q_{s_i}(x)$ and $q_{s_i}'(x)$;\
\FOR{each $(t_1,p,k_1)\in[0,\min(3\sqrt{k},t)]\times[0,\min(3\sqrt{k},t)]\times[0,k]$}
\STATE calculate $\Psi(t_1,p,k_1)$;\
\IF{$\Psi(t_1,p,k_1)=\mathtt{NULL}$} 
\STATE skip this combination of $(i,t_1,p,k_1)$;
\ELSE 
\STATE let $\mathtt{s.temp}=\Psi(t_1,p,k_1)$;
\STATE calculate the explicit form of $q_{\mathtt{s.temp}}(x)$ and $q_{\mathtt{s.temp}}'(x)$;
\STATE let $\mathtt{Cross.point}=\{(x^*,y^*):q_{s_i}(x^*)=q_{\mathtt{s.temp}}(x^*)=y^*,q_{s_i}'(x^*)\cdot q_{\mathtt{s.temp}}'(x^*)<0\}$;
\IF{$\mathtt{Cross.point}=\emptyset$} 
\STATE skip this combination of $(i,t_1,p,k_1)$;
\ELSE 
\STATE add a new row $(y^*,x^*,i,t_1,p,k_1)$ to $\mathtt{Record.matrix}$;
\ENDIF 
\ENDIF 
\ENDFOR
\ENDFOR
\STATE
\STATE Find the row (randomly pick up one if multi rows) of $\mathtt{Record.matrix}$ which has the largest $y^*$, denoted by $(y^*,x^*,i^*,t^*,p^*,k^*)$. Then, we have ${\cal R}_2(k,t)=\langle s^-\rangle\cup\langle s^+\rangle$ where $s^-=s_{i^*}$ and $s^+=\Psi(t^*,p^*,k^*)$.
\end{algorithmic}
\end{algorithm}
\fi

\section{Examples}\label{sec:example}
This section illustrates the applications of theorems in Sections \ref{sec:general} and \ref{sec:sequence}. We consider two forms of $\Sigma$. One is the identity matrix which is of type-H, so that Theorems \ref{thm:table} and \ref{thm:11} are directly applicable to derive ${\cal T}$ or reasonable subsets of sequences. The other is in the AR(1) form $\Sigma_1=(0.2^{|i-j|})_{1\leq i,j\leq k}$ which is not of type-H, hence 
we need to address ${\cal T}$ through Algorithm \ref{alg:T}. Examples \ref{example1} and \ref{example2} consider these two types of $\Sigma$ for flexible choices of $k$ and $t$. Example \ref{example3} further shows the flexibility of our method under $\Sigma_1$. These examples focus on relatively small $k$ and $t$ to save space. We devote Examples \ref{example:largekt} and \ref{example4} to large $k$ and $t$, where optimal or highly efficient designs are manually constructed according to Theorems \ref{thm:121101} and \ref{thm:11}. 
For an exact design $d$, let $0=\lambda_0\leq\lambda_1\leq\lambda_2\leq\dots\leq\lambda_{t-1}$ be all the eigenvalues of $C_d$. Its A-, D-, E-, T-efficiencies are defined as follows.
 \begin{eqnarray*}
E_A(d)=\frac{(t-1)^2}{ny^*\Sigma_{i=1}^{t-1}\lambda_i^{-1}},&& E_D(d)=\frac{t-1}{ny^*}(\Pi_{i=1}^{t-1}\lambda_i)^{\frac{1}{t-1}},\\
 E_E(d)=\frac{(t-1)\lambda_1}{ny^*}, &&E_T(d)=\frac{\Sigma_{i=1}^{t-1}\lambda_i}{ny^*}.
\end{eqnarray*}
Note the choices of $k,t,n$ are all arbitrary in this section.

\begin{example}\label{example1}
We first illustrate the application of results in {Section \ref{sec:sequence2}}, for which we take $(k,t)$ as $(5,4)$ or $(8,3)$ and assume $\Sigma=I$. Exact designs are listed in {Table \ref{tb:exactdesign}} with $n=6,15$. Note all values of $k,t,n$ are given arbitrarily and the same method applies to any other configurations as long as $t\leq 3$ or $k\leq 10$. Take $(k,t)=(5,4)$ for example, Theorem \ref{thm:table} shows that we have ${\cal T}=\langle (12341) \rangle \cup \langle (11233) \rangle$ under Model (\ref{model:3}). Exact designs are derived through the IQP based on this support. We can observe that the cases of $n=6$ and $n=15$ make different selections of sequences from ${\cal T}$. All exact designs are highly efficient. 
\begin{center}
\begin{table}
\caption{Some exact designs and their A- and D- efficiencies, $\Sigma=I$. 
}\label{tb:exactdesign}
\begin{tabular}{|c|c|c|c|}
  \hline
  $k,t,n,$ model &exact design& A-effi & D-effi  \\ \hline
    5, 4, 6, (\ref{model:3}) & $(12431)$, $(24133)\times 2$, $(34421)$, $(41223)$, $(44321)$ & 0.9868 & 0.9903 \\\hline
    5, 4, 15, (\ref{model:3}) & $(12243)\times 2$, $(12431)\times 2$, $(23144)\times 2$, $(23341)$, $(34412)$   & 0.9983 & 0.9987 \\
    & $(24133)\times 3$, $(41223)$, $(42114)$, $(43211)$, $(43221)$  &  &  \\\hline
  8, 3, 6, (\ref{model:3}) & $(12223311)$, $(22213132)$, $(33311122)\times 2$, $(33322111)\times 2$   & 0.9585 & 0.9706 \\\hline
       8, 3, 15, (\ref{model:3}) & $(11232311)$, $(22111332)\times 4$, $(22213132)$, $(23331122)\times 4$  & 0.9994 & 0.9995 \\
    & $(33311122)\times 2$, $(33312123)$, $(33322111)\times 2$   &  &  \\\hline

  5, 4, 6, (\ref{eqn:modelcrossover}) & $(12241)$, $(13344)$, $(14421)$, $(23342)$, $(33114)$, $(44322)$  & 0.9926 & 1.0000  \\\hline
  5, 4, 15, (\ref{eqn:modelcrossover}) & $(12241)$, $(12441)$, $(14421)$, $(21142)$, $(21332)\times 2$, $(23342)$   & 0.9982 & 0.9982 \\
    & $(33114)\times 2$, $(33144)$, $(33211)$, $(44122)\times 2$, $(44233)\times 2$   &  &  \\\hline
      8, 3, 6, (\ref{eqn:modelcrossover}) & $(11333221)$, $(21113322)\times 2$, $(22333112)\times 2$, $(33111223)$   & 1.0000 & 1.0000 \\\hline
      8, 3, 15, (\ref{eqn:modelcrossover}) & $(11333221)\times 2$, $(21113322)\times 5$, $(22333112)\times 5$  & 0.9994 & 0.9994 \\
    & $(33111223)\times 3$   &  &  \\\hline
    \end{tabular}
\end{table}
\end{center}
\end{example}

\begin{example}\label{example2}
Consider $\Sigma=\Sigma_1$ and the same values of $k,t,n$ as in Example \ref{example1}. Since $\Sigma_1$ is not of type-H, Theorem \ref{thm:table} is not applicable here and we shall apply {Algorithm \ref{alg:T}} to obtain ${\cal T}$ directly. Exact designs are again derived by applying IQP to ${\cal T}$ and displayed in Table \ref{tb:exactdesign3}. They are all highly efficient. 
\begin{center}
\begin{table}
\caption{Some exact designs and their A- and D- efficiencies, $\Sigma=\Sigma_1$.}\label{tb:exactdesign3}
\begin{tabular}{|c|c|c|c|}
  \hline
  $k,t,n$, model &exact design& ~A-effi~ & ~D-effi~  \\ \hline
    5, 4, 6, (\ref{model:3}) & $(11443)$, $(12234)$, $(22134)$, $(23341)$, $(32411)$, $(33244)$ & 0.9786 & 0.9816 \\\hline
  5, 4, 15, (\ref{model:3}) & $(11234)$, $(11423)$, $(12243)$, $(14322)$, $(22143)$, $(22411)$   & 0.9936 & 0.9941 \\
    & $(31123)$, $(31442)$, $(32143)$, $(32411)$, $(34213)$, $(34223)$   &  &  \\
        & $(42134)$, $(42314)$, $(43124)$   &  &  \\\hline
  8, 3, 6, (\ref{model:3}) & $(11122233)$, $(11333222)$, $(22111333)$, $(22233311)$  & 0.9857 & 0.9857 \\
   & $(33311122)$, $(22233311)$  &  &  \\\hline
       8, 3, 15, (\ref{model:3}) & $(11323231)$, $(11333221)$, $(12233311)$, $(21133322)\times3$  & 0.9979 & 0.9982 \\
    & $(22111332)\times3$, $(22333112)\times 2$, $(23311122)\times2$    &  &  \\
        & $(31212133)$, $(33111223)$   &  &  \\\hline
  5, 4, 6, (\ref{eqn:modelcrossover}) & $(12241)$, $(13344)$, $(21144)$, $(23342)$, $(31143)$, $(32244)$  & 0.9949 & 0.9994  \\\hline
  5, 4, 15, (\ref{eqn:modelcrossover}) & $(12233)$, $(13322)$, $(13341)$, $(21142)$, $(23312)$, $(24432)$   & 0.9986 & 0.9986 \\
    & $(31143)$, $(31144)$, $(32244)$, $(34411)$, $(34413)$, $(34422)$   &  &  \\
        & $(41122)$, $(42211)$, $(43314)$   &  &  \\\hline
      8, 3, 6, (\ref{eqn:modelcrossover}) & $(12223311)$, $(13332211)$, $(21113322)$, $(23331122)$   & 1.0000 & 1.0000 \\ 
      & $(31112233)$, $(32221133)$   &  &  \\\hline
      8, 3, 15, (\ref{eqn:modelcrossover}) & $(12223311)\times 2$, $(13332211)\times 3$, $(23311122)\times 3$  & 0.9997 & 0.9997 \\
    & $(23331122)\times 2$, $(31112233)\times 2$, $(32221133)\times 3$   &  &  \\\hline
\end{tabular}
\end{table}
\end{center}
\end{example}


\begin{example}\label{example3}
Both of the above examples only examined two values of $n$. Here we consider a continous spectrum of $n$ values with other parameters being the same as in Example \ref{example2}. For this purpose, we carry out a series of calculation like Example \ref{example2} for all $n\in[8,50]$ when $k = 5$ and  $t = 3$. The A-, D-, E- and T- efficiencies of these designs are plotted against $n$ in Figures \ref{penG} and \ref{penG3} for Models (\ref{model:3}) and (\ref{eqn:modelcrossover}), respectively. Also see Figure \ref{penG2} for a more complex case of $k=11$ and $t=4$ under Model (\ref{model:3}).

\begin{figure}
\center
\includegraphics[width=6in]{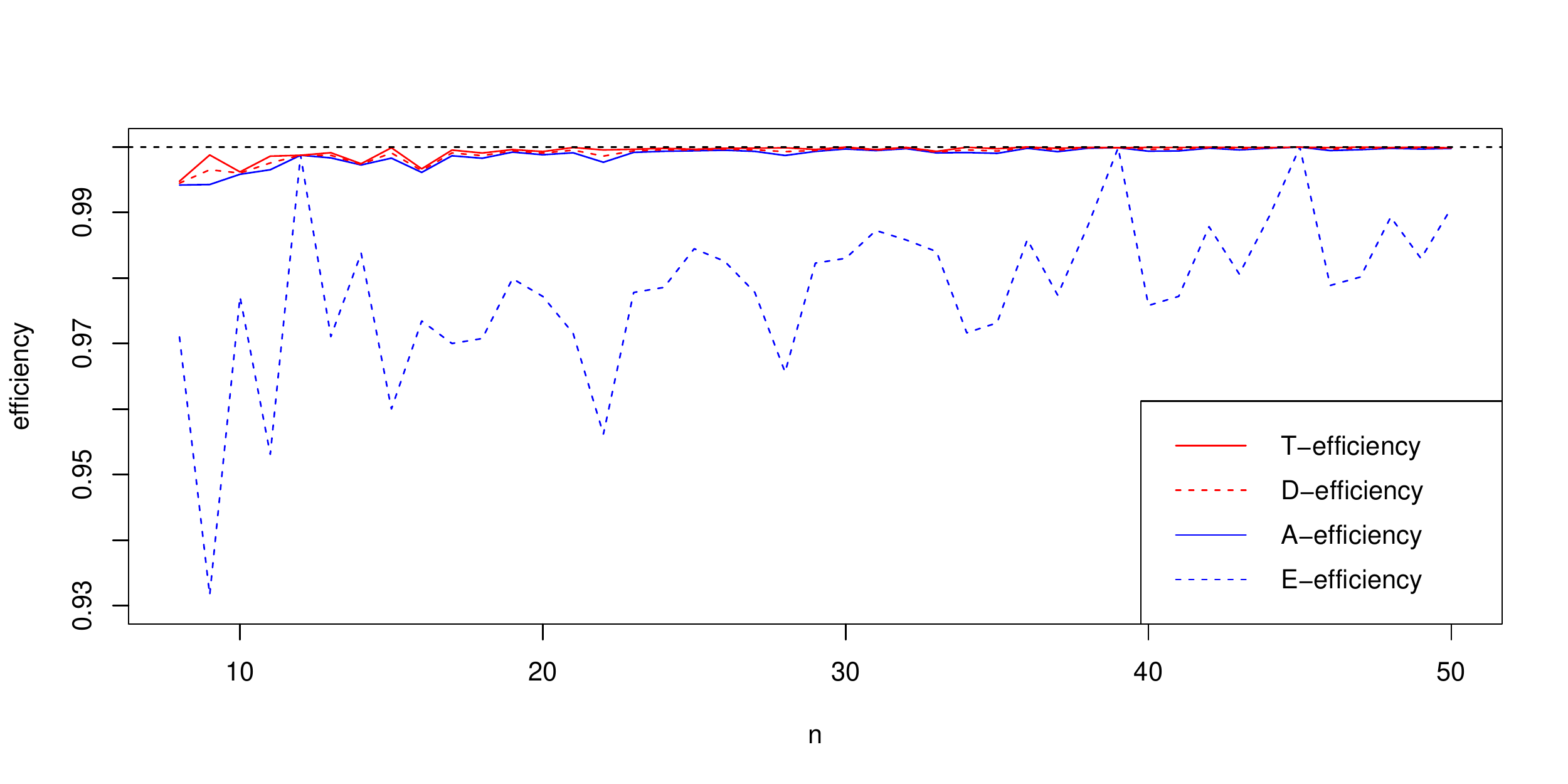}
\caption{A-, D-, T-, E- efficiencies of exact designs for different $n$ with $(k,t)=(5,3)$ under Model (\ref{model:3}).}
\label{penG}
\end{figure}

\begin{figure}
\center
\includegraphics[width=6in]{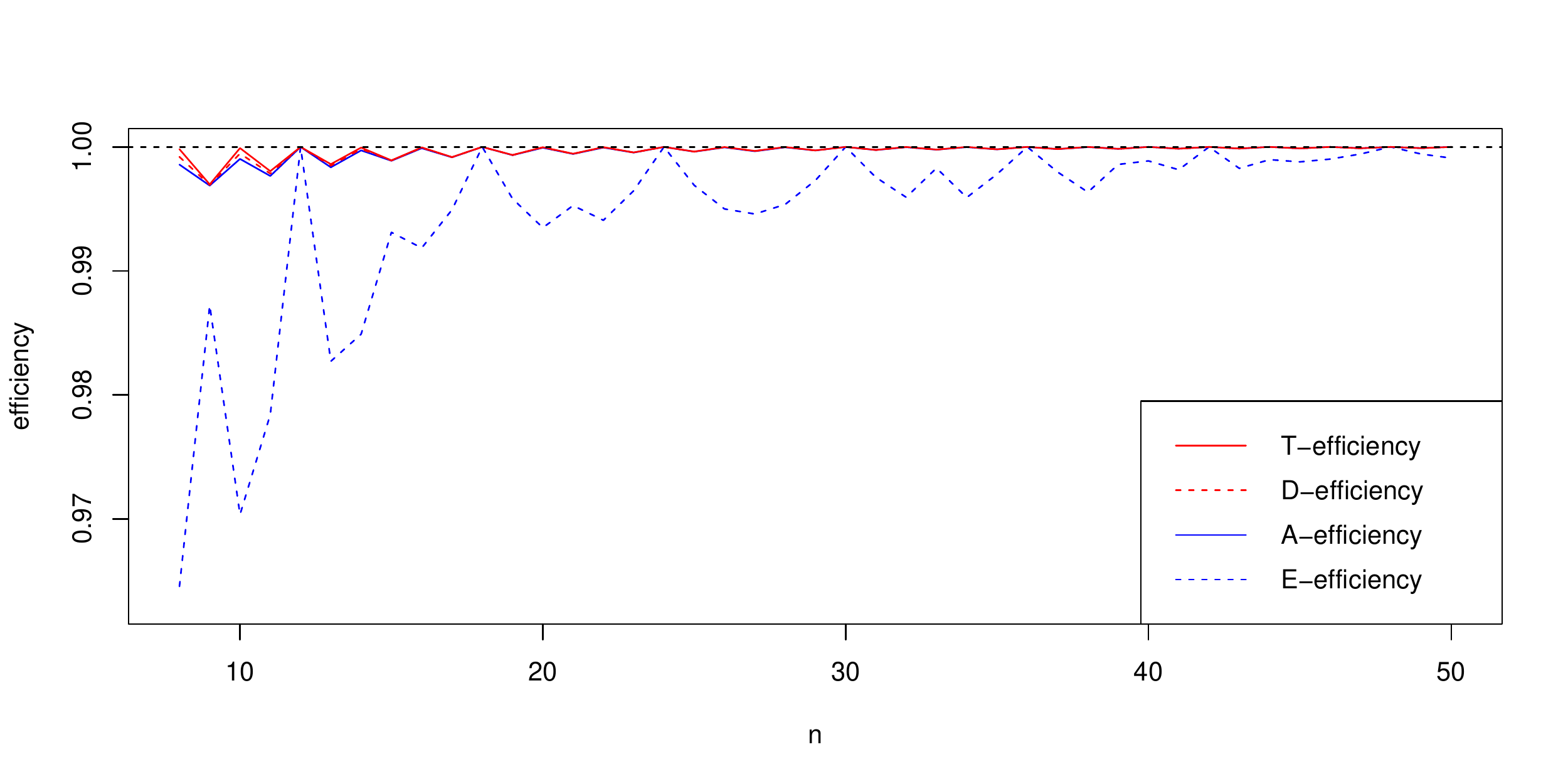}
\caption{A-, D-, E-, T- efficiencies of exact designs for different $n$ with $(k,t)=(5,3)$ under the crossover Model (\ref{eqn:modelcrossover}).}
\label{penG3}
\end{figure}

\begin{figure}
\center
\includegraphics[width=6in]{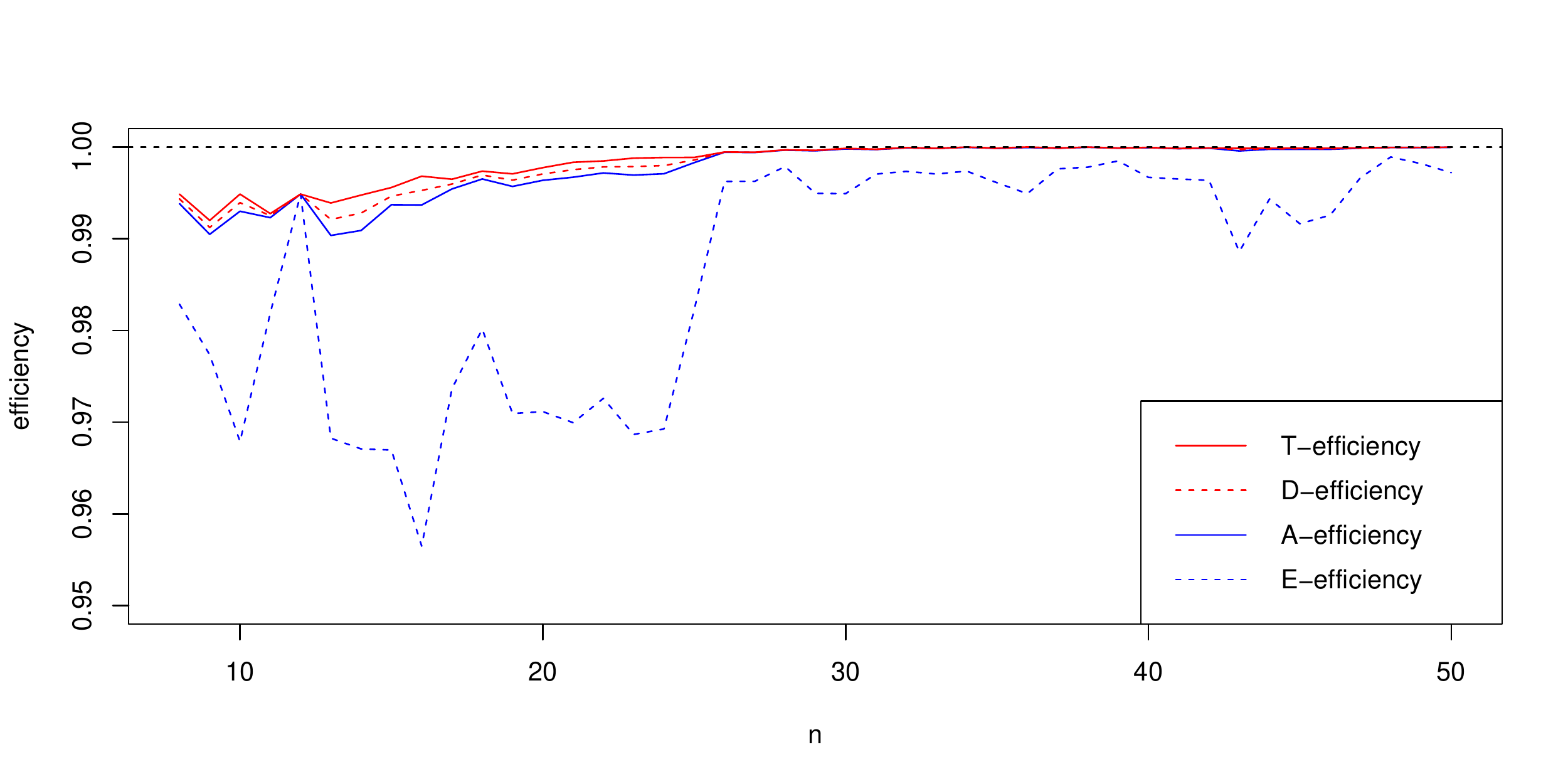}
\caption{A-, D-, E-, T- efficiencies of exact designs for different $n$ with $(k,t)=(11,4)$ under Model (\ref{model:3}).}
\label{penG2}
\end{figure}
\end{example}

In the above three examples, we have restricted to $k\leq 10$ so that Theorem \ref{thm:table} can directly display ${\cal T}$ and/or Algorithm \ref{alg:T} can calculate ${\cal T}$ within manageable time. Next, we shall illustrate how to obtain strictly universally optimal approximate designs when $k\geq 11$. Generating corresponding optimal exact designs is a simple task since (\ref{eqn:227}) can be solved easily over only two equivalence classes.

\begin{example}\label{example:largekt} For $11\leq k\leq 50$, $t=8$ and $\Sigma=I$, we list obtained designs in Tables \ref{tb:largekt2} and  \ref{tb:largekt3} (in supplementary materials). We illustrate the process of deriving the designs for the case of $(k,t)=(31,8)$ here. For convenience, we use $z_q$ to represent a subsequence $(z,\ldots,z)$ of length $q$. The ${\cal S}^*$ in Theorem \ref{lem:121001} contains $1116$ different sequences. Then we can restrict Algorithm \ref{alg:T} to ${\cal S}^*$ and it takes less than 0.3 second to obtain ${\cal C}^*=\{s_1,s_2\}$, where $s_1=(1_6,2_5,3_5,4_5,5_5,6_5)$ and $s_2=({\cal M}(1_4,2_3),3_6,4_6,5_6,6_6)$. The other sequences in the pseudo equivalence class ${\cal C}(s_1)$ are obtained by reducing the replication of treatment $1$ down to $5$ and let another treatment to be replicated $6$ times. However, all these sequences have the same combinatorial features from the perspective of design construction, hence it is sufficient to only consider one sequence, e.g. $s_1$. It is the same case for $s_2$ even though we have $5$ sequences in ${\cal C}(s_2)$. Hence by theorem Theorem \ref{thm:121101}, we shall have $11$ equivalence classes in ${\cal T}$, with $6$ from ${\cal C}(s_1)$ and $5$ from ${\cal C}(s_2)$. However, we can construct optimal designs on $\langle s_1\rangle$ and $\langle s_2\rangle$.


\end{example}

In Example \ref{example:largekt}, we can see that $s_1$ in Table \ref{tb:largekt2} (supplementary materials) always has a dominating proportion in optimal designs. In fact, it is the same sequence in Theorem \ref{thm:11} that produces highly efficient designs. We illustrate this process in the following Example \ref{example4}.

\begin{example}\label{example4}
Consider $k=11$, $t=5$ and $\Sigma=I$. Theorem \ref{thm:11} suggests the set of sequences $\langle s_3\rangle$ with $s_3=(11112222333)$. One can generate a  symmetric exact design based on $s_3$ as follows. Find an $OA_I$ with $3$ rows and $5$ levels, denoted by ${M}$.
\begin{equation}
{M}=
\left(
\begin{array}{cccccccccccccccccccccc}
1   &~ 5   &~ 5&~    2  &~  5&~    1   &~ 2&~    4   &~ 1 &~    4  &~   4 &~    3   &~  1 &~    4   &~  3  &~   3 &~    3 &~    5 &~   2  &~   2\\
4   &~ 1 &~   3   &~ 5&~    2   &~ 5 &~   3   &~ 3 &~   3    &~ 1&~     5   &~  4 &~    2   &~  2  &~   1 &~    5&~     2 &~    4 &~    1   &~  4\\
5    &~2  &~  1  &~  4 &~   4  &~  2  &~  5  &~  2   &~ 4   &~  5  &~   3  &~   2  &~   3   &~  1  &~  4   &~  1   &~  5   &~  3   &~  3   &~  1
\end{array}
\right).\notag
\end{equation}
By definition of $OA_I$, each pair of distinct symbols from ${\mathbb Z}_5$ appears equally often, exactly once here, for all three $2\times 20$ subarrays of ${M}$. The first column  $(1,4,5)$ means that we shall replace the symbols $1,2,3$ in $s_3$ by $1,4,5$ respectively, and hence produce a new sequence from $s_3$ as $(11114444555)$. Applying all columns of ${M}$ to $s_3$ results in a new array
\begin{equation}
{M}_{s_3}=
\left(
\begin{array}{cccccccccccccccccccccc}
1   &~ 5   &~ 5&~    2  &~  5&~    1   &~ 2&~    4   &~ 1 &~    4  &~   4 &~    3   &~  1 &~    4   &~  3  &~   3 &~    3 &~    5 &~   2  &~   2\\
1   &~ 5   &~ 5&~    2  &~  5&~    1   &~ 2&~    4   &~ 1 &~    4  &~   4 &~    3   &~  1 &~    4   &~  3  &~   3 &~    3 &~    5 &~   2  &~   2\\
1   &~ 5   &~ 5&~    2  &~  5&~    1   &~ 2&~    4   &~ 1 &~    4  &~   4 &~    3   &~  1 &~    4   &~  3  &~   3 &~    3 &~    5 &~   2  &~   2\\
1   &~ 5   &~ 5&~    2  &~  5&~    1   &~ 2&~    4   &~ 1 &~    4  &~   4 &~    3   &~  1 &~    4   &~  3  &~   3 &~    3 &~    5 &~   2  &~   2\\
4   &~ 1 &~   3   &~ 5&~    2   &~ 5 &~   3   &~ 3 &~   3    &~ 1&~     5   &~  4 &~    2   &~  2  &~   1 &~    5&~     2 &~    4 &~    1   &~  4\\
4   &~ 1 &~   3   &~ 5&~    2   &~ 5 &~   3   &~ 3 &~   3    &~ 1&~     5   &~  4 &~    2   &~  2  &~   1 &~    5&~     2 &~    4 &~    1   &~  4\\
4   &~ 1 &~   3   &~ 5&~    2   &~ 5 &~   3   &~ 3 &~   3    &~ 1&~     5   &~  4 &~    2   &~  2  &~   1 &~    5&~     2 &~    4 &~    1   &~  4\\
4   &~ 1 &~   3   &~ 5&~    2   &~ 5 &~   3   &~ 3 &~   3    &~ 1&~     5   &~  4 &~    2   &~  2  &~   1 &~    5&~     2 &~    4 &~    1   &~  4\\
5    &~2  &~  1  &~  4 &~   4  &~  2  &~  5  &~  2   &~ 4   &~  5  &~   3  &~   2  &~   3   &~  1  &~  4   &~  1   &~  5   &~  3   &~  3   &~  1\\
5    &~2  &~  1  &~  4 &~   4  &~  2  &~  5  &~  2   &~ 4   &~  5  &~   3  &~   2  &~   3   &~  1  &~  4   &~  1   &~  5   &~  3   &~  3   &~  1\\
5    &~2  &~  1  &~  4 &~   4  &~  2  &~  5  &~  2   &~ 4   &~  5  &~   3  &~   2  &~   3   &~  1  &~  4   &~  1   &~  5   &~  3   &~  3   &~  1
\end{array}
\right).\notag
\end{equation}
The design ${M}_{s_3}$ is symmetric and has the same values of A-, D-, E-, and T-efficiencies, with its columns as blocks. These alphabetical efficiencies are all equal to the efficiency of ${\langle s_3\rangle}$, particularly we have $e_{\langle s_3\rangle}=0.9862$ for both Models (\ref{model:3}) and (\ref{model:2}). For $k=11$ and $t\in\{5,7,11,23\}$, the representative sequence and the design efficiency is unchanged among different choices of $t$. The only difference is that the minimum number of blocks in generating symmetric design is $20,42,110$ and $506$, respectively.

For larger values of $(k,t)$,
let us consider $k=37$, $t\in \{5,7,11,23\}$ and $\Sigma=I$. Theorem \ref{thm:11} suggests we should adopt the representative sequence $(1111111122222222333333344444$ $445555555)$ for both Models (\ref{model:3}) and (\ref{model:2}) and for all $t$ under consideration. The efficiencies of the derived designs are $0.9997,0.9992,0.9992,0.9992$ for $n$ as $20,42,110,506$, respectively.
\end{example}

\bigskip
\begin{center}
{\large\bf SUPPLEMENTARY MATERIAL}
\end{center}

\begin{description}
\item[Title:] Circular optimal design supplementary material. (PDF)
\end{description}

\newpage

\fontsize{12}{14pt plus.8pt minus .6pt}\selectfont \vspace{0.8pc}
\markboth{\hfill{\footnotesize\rm XIANGSHUN KONG AND WEI ZHENG} \hfill}
{\hfill {\footnotesize\rm SUPPLEMENTARY MATERIAL} \hfill}
\centerline{\large\bf Optimal and efficient circular designs
with neighboring effects}
\centerline{\large\bf (supplementary material)}
\fontsize{12}{14pt plus.8pt minus .6pt}\selectfont \vspace{0.2pc}

\begin{center}
\begin{table}
\center
\caption{True optimal designs in Example \ref{example:largekt} with $11\leq k\leq 30$ and $t=8$. For convenience, we use $z_q$ to represent an subsequence $(z,\ldots,z)$ of length $q$. $p_{\langle s_2 \rangle}$ is omitted since $p_{\langle s_2 \rangle}=1-p_{\langle s_1 \rangle}$. We refer to (\ref{eqn:121001})--(\ref{eqn:121003}) for the definition of ${\cal M}$.}\label{tb:largekt2}
\begin{tabular}{|c|c|c|c|c|c|c|c|c|c|c|ccccccccccccc|}\hline
$k$&$s_1$&$s_2$&$p_{\langle s_1 \rangle}$\\
  \hline
    $11$&$(1_4,2_4,3_3)$&$(~{\cal M}(1_2,2_2),3_4,4_3~)$&$0.8034$\\\hline
  $12$&$(1_3,2_3,3_3,4_3)$&$(~{\cal M}(1_2,2_2),{\cal M}(3_2,4_2),{\cal M}(5_2,6_2)~)$&$0.9264$\\\hline
  $13$&$(1_4,2_3,3_3,4_3)$&$(~{\cal M}(1_3,2_2),3_4,4_4~)$&$0.8514$\\\hline
  $14$&$(1_4,2_4,3_3,4_3)$&$(~{\cal M}(1_3,2_3),3_4,4_4~)$&$0.8889$\\\hline
  $15$&$(1_4,2_4,3_4,4_3)$&$(~{\cal M}(1_4,2_3),3_4,4_4~)$&$0.9053$\\\hline
  $16$&$(1_4,2_4,3_4,4_4)$&$(~{\cal M}(1_3,2_3),{\cal M}(3_3,4_3),5_4~)$&$0.9529$\\\hline
  $17$&$(1_5,2_4,3_4,4_4)$&$(~{\cal M}(1_3,2_2),3_4,4_4,5_4~)$&$0.8784$\\\hline
  $18$&$(1_5,2_5,3_4,4_4)$&$(~{\cal M}(1_3,2_3),3_4,4_4,5_4~)$&$0.9017$\\\hline
  $19$&$(1_5,2_5,3_5,4_4)$&$(~{\cal M}(1_3,2_3),3_5,4_4,5_4~)$&$0.9088$\\\hline
  $20$&$(1_5,2_5,3_5,4_5)$&$(~{\cal M}(1_3,2_3),3_5,4_5,5_4~)$&$0.9150$\\\hline
  $21$&$(1_5,2_4,3_4,4_4,5_4)$&$(~{\cal M}(1_3,2_3),3_5,4_5,5_5~)$&$0.8956$\\\hline
  $22$&$(1_5,2_5,3_4,4_4,5_4)$&$(~{\cal M}(1_3,2_3),{\cal M}(3_3,4_3),5_5,6_5~)$&$0.9500$\\\hline
  $23$&$(1_5,2_5,3_5,4_4,5_4)$&$(~{\cal M}(1_3,2_3),{\cal M}(3_3,4_3),{\cal M}(5_3,6_3),7_5~)$&$0.9684$\\\hline
  $24$&$(1_5,2_5,3_5,4_5,5_4)$&$(~{\cal M}(1_3,2_3),{\cal M}(3_3,4_3),{\cal M}(5_3,6_3),{\cal M}(7_3,8_3)~)$&$0.9775$\\\hline
  $25$&$(1_5,2_5,3_5,4_5,5_5)$&$(~{\cal M}(1_4,2_3),{\cal M}(3_3,4_3),{\cal M}(5_3,6_3),{\cal M}(7_3,8_3)~)$&$0.9798$\\\hline
  $26$&$(1_6,2_5,3_5,4_5,5_5)$&$(~{\cal M}(1_3,2_3),3_5,4_5,5_5,6_5~)$&$0.9120$\\\hline
  $27$&$(1_6,2_6,3_5,4_5,5_5)$&$(~{\cal M}(1_4,2_3),3_5,4_5,5_5,6_5~)$&$0.9285$\\\hline
  $28$&$(1_6,2_6,3_6,4_5,5_5)$&$(~{\cal M}(1_4,2_4),3_5,4_5,5_5,6_5~)$&$0.9405$\\\hline
  $29$&$(1_6,2_6,3_6,4_6,5_5)$&$(~{\cal M}(1_4,2_3),3_6,4_6,5_5,6_5~)$&$0.9346$\\\hline
  $30$&$(1_6,2_6,3_6,4_6,5_6)$&$(~{\cal M}(1_4,2_4),{\cal M}(3_4,4_3),5_5,6_5,7_5~)$&$0.9708$\\\hline
\end{tabular}
\end{table}
\end{center}

\begin{center}
\begin{table}
\center
\caption{True optimal designs in Example \ref{example:largekt} with $31\leq k\leq 50$ and $t=8$. For convenience, we use $z_q$ to represent an subsequence $(z,\ldots,z)$ of length $q$. $p_{\langle s_2 \rangle}$ is omitted since $p_{\langle s_2 \rangle}=1-p_{\langle s_1 \rangle}$. We refer to (\ref{eqn:121001})--(\ref{eqn:121003}) for the definition of ${\cal M}$.}\label{tb:largekt3}
\begin{tabular}{|c|c|c|c|c|c|c|c|c|c|c|ccccccccccccc|}\hline
$k$&$s_1$&$s_2$&$p_{\langle s_1 \rangle}$\\
  \hline
  $31$&$(1_6,2_5,3_5,4_5,5_5,6_5)$&$(~{\cal M}(1_4,2_3),3_6,4_6,5_6,6_6~)$&$0.9252$\\\hline
  $32$&$(1_6,2_6,3_5,4_5,5_5,6_5)$&$(~{\cal M}(1_4,2_4),3_6,4_6,5_6,6_6~)$&$0.9348$\\\hline
  $33$&$(1_6,2_6,3_6,4_5,5_5,6_5)$&$(~{\cal M}(1_4,2_4),{\cal M}(3_4,4_3),5_6,6_6,7_6~)$&$0.9667$\\\hline
  $34$&$(1_6,2_6,3_6,4_6,5_5,6_5)$&$(~{\cal M}(1_4,2_4),{\cal M}(3_4,4_4),5_6,6_6,7_6~)$&$0.9695$\\\hline
  $35$&$(1_6,2_6,3_6,4_6,5_6,6_5)$&$(~{\cal M}(1_4,2_4),{\cal M}(3_4,4_4),{\cal M}(5_4,6_3),7_6,8_6~)$&$0.9796$\\\hline
  $36$&$(1_6,2_6,3_6,4_6,5_6,6_6)$&$(~{\cal M}(1_4,2_4),{\cal M}(3_4,4_4),{\cal M}(5_4,6_4),7_6,8_6~)$&$0.9810$\\\hline
  $37$&$(1_7,2_6,3_6,4_6,5_6,6_6)$&$(~{\cal M}(1_4,2_4),{\cal M}(3_4,4_4),{\cal M}(5_4,6_4),7_7,8_6~)$&$0.9813$\\\hline
  $38$&$(1_7,2_7,3_6,4_6,5_6,6_6)$&$(~{\cal M}(1_4,2_4),3_6,4_6,5_6,6_6,7_6~)$&$0.9416$\\\hline
  $39$&$(1_7,2_7,3_7,4_6,5_6,6_6)$&$(~{\cal M}(1_4,2_4),3_7,4_6,5_6,6_6,7_6~)$&$0.9434$\\\hline
  $40$&$(1_7,2_7,3_7,4_7,5_6,6_6)$&$(~{\cal M}(1_5,2_4),3_7,4_6,5_6,6_6,7_6~)$&$0.9524$\\\hline
  $41$&$(1_7,2_7,3_7,4_7,5_7,6_6)$&$(~{\cal M}(1_4,2_4),{\cal M}(3_4,4_4),5_7,6_6,7_6,8_6~)$&$0.9733$\\\hline
  $42$&$(1_7,2_7,3_7,4_7,5_7,6_7)$&$(~{\cal M}(1_4,2_4),{\cal M}(3_4,4_4),5_7,6_7,7_6,8_6~)$&$0.9741$\\\hline
  $43$&$(1_7,2_6,3_6,4_6,5_6,6_6,7_6)$&$(~{\cal M}(1_4,2_4),3_7,4_7,5_7,6_7,7_7~)$&$0.9397$\\\hline
  $44$&$(1_7,2_7,3_6,4_6,5_6,6_6,7_6)$&$(~{\cal M}(1_5,2_4),3_7,4_7,5_7,6_7,7_7~)$&$0.9476$\\\hline
  $45$&$(1_7,2_7,3_7,4_6,5_6,6_6,7_6)$&$(~{\cal M}(1_5,2_5),3_7,4_7,5_7,6_7,7_7~)$&$0.9539$\\\hline
  $46$&$(1_7,2_7,3_7,4_7,5_6,6_6,7_6)$&$(~{\cal M}(1_5,2_5),{\cal M}(3_4,4_4),5_7,6_7,7_7,8_7~)$&$0.9750$\\\hline
  $47$&$(1_7,2_7,3_7,4_7,5_7,6_6,7_6)$&$(~{\cal M}(1_5,2_5),{\cal M}(3_5,4_4),5_7,6_7,7_7,8_7~)$&$0.9769$\\\hline
  $48$&$(1_7,2_7,3_7,4_7,5_7,6_7,7_6)$&$(~{\cal M}(1_5,2_5),{\cal M}(3_5,4_5),5_7,6_7,7_7,8_7~)$&$0.9785$\\\hline
  $49$&$(1_7,2_7,3_7,4_7,5_7,6_7,7_7)$&$(~{\cal M}(1_5,2_5),{\cal M}(3_5,4_5),5_8,6_7,7_7,8_7~)$&$0.9793$\\\hline
  $50$&$(1_8,2_7,3_7,4_7,5_7,6_7,7_7)$&$(~{\cal M}(1_4,2_4),3_7,4_7,5_7,6_7,7_7,8_7~)$&$0.9469$\\\hline
\end{tabular}
\end{table}
\end{center}

This appendix provides the proof of Theorems \ref{prop:pd}, \ref{thm:1}, \ref{thm:3}, \ref{thm:3.4}, Theorems \ref{thm:5}, \ref{thm:table}, \ref{lem:121001}, \ref{thm:121101}, \ref{thm:11}.
The following definition and notations are frequently used in most proofs in this appendix.
For sequence $s=(t_1\dots t_k)$, let $\gamma_s=\sum_{i=1}^{k-1}\delta_{t_i,t_{i+1}}$, $\psi_s=\sum_{i=2}^{k-1}\delta_{t_{i-1},t_{i+1}}$, $f_{s, j}=\sum_{i=1}^{k}\delta_{t_i,j}$, $\chi_s=\sum_{i=1}^{t}f_{s, i}^2$. Here $\delta_{ij}$ is the Kronecker delta. 
In the beginning of Section \ref{sec:formulation}, we restrict our analysis to $k\geq 4$. That is because $C_d=0$ when $k\leq 3$, which means all contrasts of $\phi$ are not estimable. See Proposition \ref{thm:122001}.

\begin{proposition}\label{thm:122001}
For $k\leq 3$, we have $C_d=0$ for both interference models (\ref{model:3}) and (\ref{model:2}).
\end{proposition}
\begin{proof}
Let $N$ be the $n\times t$ block-treatment incident matrix so that its $(i,j)$-th entry is given by the number of times that treatment $j$ appears in block $i$. When $k=3$, one can verify that $3T_d=UN-(L_d-T_d+R_d-T_d)$. When $k=2$, we have $2T_d = UN-(L_d-T_d + R_d-T_d)/2$. The lemma is concluded by (3) and the definition of the projection operator with $C_d=T_d'Vpr^{\perp}\{VU|V(L_d-T_d)|V(R_d-T_d)\}VT_d$, where $V$ is a symmetric matrix such that $V^2=I_n\otimes \Sigma^{-1}$.
\end{proof}


\begin{lemma}\label{thm:112901}
($i$) $y_{\xi}=y^*$ implies $q_{\xi}(x^*)=y^*$ for any measure $\xi\in {\cal P}_0$; ($ii$) $y_{\xi}=y^*$ implies ${\cal V}_{\xi}\subset {\cal T}$ for any measure $\xi\in {\cal P}_0$.
\end{lemma}
\begin{proof}
This can be proved analogously to Proposition 1 in \cite{zheng:2015}.
\end{proof}

\vspace{.4cm}

{\noindent\bf Proof of Theorem \ref{prop:pd}.} Consider the more complex directional Model (\ref{model:3}). The parallel result can be proved analogously for undirectional Model (\ref{model:2}). Review $Q_s=(c_{ij})_{i,j=1,2}$, where $c_{ij}=tr(B_tC_{ij}B_t)$, $C_{sij}=G_i'\tilde{ B}G_j$, $0\leq i,j\leq 2$, with $G_0=T_s$, $G_1=\tilde{L}_s=L_s-T_s$, $G_2=\tilde{R}_s=R_s-T_s$, and $\tilde{B}=\Sigma^{-1}-\Sigma^{-1}J_k\Sigma^{-1}/1^{\prime}_k\Sigma^{-1}1_k$ with $J_k=1_k1_k'$. Rewrite $\tilde{B}=\Sigma^{-1/2}(I_n-\Sigma^{-1/2}1_k1_k'\Sigma^{-1/2}/1_k'\Sigma^{-1}1_k)\Sigma^{-1/2}=\hat{B}'\hat{B}$. Thus, 
\begin{eqnarray*}\det(Q_s)=tr({\cal M}'_1{\cal M}_1)tr({\cal M}'_2{\cal M}_2)-tr^2({\cal M}'_1{\cal M}_2)
\end{eqnarray*} 
for ${\cal M}_i=(m_{ijk})_{j,k}=\hat{B}G_iB_t=\hat{B}G_i$ since $\hat{B}G_i1_t=0$ with $i=1,2$. Simple calculation reveals that
\begin{eqnarray*}
\det(Q_s)=\bigg(\sum_{jk}m^2_{1jk}\bigg)\bigg(\sum_{jk}m^2_{2jk}\bigg)-\bigg(\sum_{jk}m_{1jk}m_{2jk}\bigg)^2.
\end{eqnarray*}
We have $\det(Q_s)\geq 0$ and $\det(Q_s)=0$ if and only if 
\begin{eqnarray*}
m_{1j_1k_1}^2m_{2j_2k_2}^2-2m_{1j_1k_1}m_{2j_2k_2}m_{1j_2k_2}m_{2j_1k_1}+m^2_{1j_2k_2}m^2_{2j_1k_1}=0,
\end{eqnarray*}
for any $(j_1,k_1)\neq (j_2,k_2)$, which indicates $m_{1j_1k_1}m_{2j_2k_2}=m_{1j_2k_2}m_{2j_1k_1}$. 
If there exists a $(t_1,k_1)$ such that $m_{1i_1j_1}\neq 0$ and $m_{2i_1j_1}\neq 0$, then we have
\begin{eqnarray*}
m_{1j_1k_1}\cdot m_{2j_2k_2}=m_{2j_1k_1}\cdot m_{1j_2k_2}
\end{eqnarray*}
for arbitrary $(j_2,k_2)$, which indicates ${\cal M}_1=(m_{1j_1k_1}/m_{2j_1k_1})\cdot {\cal M}_2$. When there does not exist such a $(t_1,k_1)$, then we have $m_{1jk}m_{2jk}=0$ for all $(j,k)$. With the assumption that $\det(Q_s)=0$, we have
\begin{eqnarray*}
\bigg(\sum_{jk}m^2_{1jk}\bigg)\bigg(\sum_{jk}m^2_{2jk}\bigg)=\det(Q_s)+\bigg(\sum_{jk}m_{1jk}m_{2jk}\bigg)^2=0,
\end{eqnarray*}
which means ${\cal M}_1=0$ or ${\cal M}_2=0$. ${\cal M}_1$ or ${\cal M}_2$ equals to 0 if and only if the columns of $G_1$ or $G_2$ are all equivalent to $(0,0,\ldots,0,0)'$ since the sum of elements in each column is 0. In this case, we can see that the $q(x)$ function is always $0$ and the sequence can only be $(1,1,1,\ldots,1)$, which further derives ${\cal M}_1={\cal M}_2=0$ and so $Q_s=0$. Thus, $s$ with $Q_s=0$ can not be a support point of the optimal design $\xi$ since if so we will have $y^*_{\xi}=0$, which means no information at all. This leads to the first conclusion in part ($i$) and the first conclusion in part ($ii$).

The only possible case to achieve $\det(Q_s)=0$ is that ${\cal M}_1=(m_{1j_1k_1}/m_{2j_1k_1})\cdot {\cal M}_2$.
Note that ${\cal M}_1=\hat{B}G_1=\hat{B}(L_s-T_s)$ and ${\cal M}_2=\hat{B}G_2=\hat{B}(R_s-T_s)$. Let $\rho_0=m_{1j_1k_1}/m_{2j_1k_1}$, we have
\begin{eqnarray*}
0&=&\hat{B}(L_s-T_s)-\rho_0\hat{B}(R_s-T_s)=\hat{B}\{(L_s-T_s)-\rho_0(R_s-T_s)\}.
\end{eqnarray*}
Since $rank(\tilde{B})=k-1$ and $\tilde{B}\cdot 1_k=0$, we have $(L_s-T_s)-\rho_0(R_s-T_s)=1_k\cdot (c_1,\ldots,c_t)$ for some $(c_1,\ldots,c_t)\in\mathbb{R}^t$.
Meanwhile, $1_k'(L_s-T_s)=1_k'(R_s-T_s)=0$ and so $1_k'\cdot 1_k\cdot(c_1,\ldots,c_t)=0$, i.e., $(c_1,\ldots,c_t)=0$ and $(L_s-T_s)-\rho_0(R_s-T_s)=0$. Write $L_s-T_s$ the form of different rows as $L_s-T_s=(r_1',\ldots,r_k')'$. Then we have $T_s-R_s=(r_2',\ldots,r_k',r_1')'$. This indicates $r_{k-1}=-\rho_0r_k$, $r_{k-2}=-\rho_0r_{k-1}$, $\ldots$ ,$r_1=-\rho_0r_2$, $r_k=-\rho_0r_1$. Now we consider the exact form of $r_k$. There are two choices for $t_k$: ($i$) all elements are $0$ or ($ii$) one element equals $1$, another one equals $-1$ and all others equal $0$. For case ($i$), we have $L_s-T_s=0$, which can be achieve only by $(1,1,\ldots, 1,1)$. For case ($ii$), we have $\rho_0$ equals to either $1$ or $-1$. When $\rho=-1$, the matrix $L_s-T_s=0$ violates the basic property that the sum of each column equals to $0$ (now there are two column sums equal $k$ or $-k$). When $\rho=1$,
from the definitions of $G_1$ and $G_2$, we can see it can be achieved if and only if $k$ is even and $s=\langle(12\ldots12)\rangle$.
This concludes the second part of ($i$).

At this point, to violate $Q_{\xi}>0$, we must have ${\cal T}\in \langle(12\ldots12)\rangle$. 
Now we shall show that the minimized value of $q_s(x)$ when $s=(12\ldots12)$ can not be $y^*$. It can be shown that $\min_xq_{s_0}(x)=tr(T'_{s_0}\tilde{B}T_{s_0})-tr[(T'_{s_0}-L'_{s_0})\tilde{B}(T_{s_0}-L_{s_0})]/4$. Let $\alpha=(10\cdots10)$ and $\beta=(01\cdots01)$.
Then simple calculation reveals $\min_xq_{s_0}(x)=\alpha \tilde{B}\alpha'+\beta \tilde{B}\beta'+(\alpha \tilde{B}\alpha'+\beta \tilde{B}\beta'-2\alpha \tilde{B}\beta')/2$. Note that $\tilde{B}\alpha=\tilde{B}(1_k-\beta)=-\tilde{B}\beta$. We have $\min_xq_{s_0}(x)=0$. Thus, it is obvious that this single equivalent class can not generate the optimal design since otherwise the information will be $0$. Thus, we have $Q_{\xi}>0$ if $\xi$ is universally optimal among $\cal P$. This concludes the second part of ($ii$).

Result ($iii$) can be proved analogously to Proposition 1 in \cite{zheng:2015}.

\vspace{.4cm}

{\noindent\bf Proof of Theorem \ref{thm:1}.} Note that $\max_{\xi \in {\cal P}}q_{\xi}(x)=r(x)$. We have $q^*_{\xi}=\min_{x\in \mathbb{R}^2} q_{\xi}(x)\le \min_{x\in \mathbb{R}^2}$ $\max_{\xi\in {\cal P}}q_{\xi}(x)=y_*$, which implies $y^*=\max_{\xi \in {\cal P}}q^*_{\xi}\leq y_*$. 
Define ${\cal T}_0=\{s:q_s(x^*)=r(x^*)\}$. If ${\cal T}_0$ contains a single sequence, it is obvious that we have $y^*\geq r(x^*)=y_*$. Now, we consider the case ${\cal T}_0$ has more than one sequence.
Let $\triangledown q_{s}(x)$ (resp. $\triangledown q_{\xi}(x)$) be the gradient of the bivariate function $q_{s}(\cdot)$ (resp. $q_{\xi}(\cdot)$) evaluated at point $x$. For convenience, we also use $\triangledown q_{s}(x)$ for the derivative function if $x\in\mathbb{R}$. For ${\cal C}_{\cal T}=\{\sum_{s\in {\cal T}}w_s\triangledown q_{s}(x^*):w_s\geq 0,\sum_{s\in {\cal T}} w_s>0\}$, we claim $0\in {\cal C}_{\cal T}$. Otherwise, there
exists a vector $c\in \mathbb{R}^2$ such that $c'\triangledown q_{s}(x)<0$ for all $s\in {\cal T}$, which implies that $r(x)$ decreases
in the direction $c$ at point $x^*$, a contradiction to the condition of $r(x)=y_*$. As a result, there exists a measure, say $\xi_0$, such that $\triangledown q_{\xi_0}(x^*)=0$ and $q_{\xi_0}(x^*)=y_*$. Then $y^*\ge q_{\xi 0}^*=
q_{\xi_0}(x^*)=y_*$. Hence part ($i$) is concluded.

Part ($ii$) naturally follows from a typical argument of the general equivalence theory (GET) for D-criterion by treating the scaler $y_{\xi}$ as the Schur complement of the matrix $F_{\xi}$. Part ($iii$) can be proved analogously to Theorem 2 in \cite{zheng:2015}.

\vspace{.4cm}

{\noindent\bf Proof of Theorem \ref{thm:3}.} Equations (\ref{eq3-1})--(\ref{eq3-3}) are equivalent to
\begin{eqnarray}
E_{\xi00}+E_{\xi00}(x^*\otimes B_t)&=&y^*B_t/(t-1),\label{proof3-1}\\
E_{\xi10}+E_{\xi11}(x^*\otimes B_t)&=&0,\label{proof3-2}\\
\sum_{s\in{\cal T}}p_s&=&1\label{proof3-3}.
\end{eqnarray}

There exists a symmetric measure denoted by $\xi_1$ such that $\xi_1\in{\cal P}^*$. From Theorem \ref{prop:pd}, we conclude $C_{\xi}=C_{\xi_1}=y^*B_t/(t-1)$. For $\xi_2=(\xi_1+\xi)/2$, we have $C_{\xi_2}\geq (C_{\xi}+C_{\xi_1})/2=y^*B_t/(t-1)$. Note that $\xi_1$ is already optimal, we have $C_{\xi_2}=y^*B_t/(t-1)$. Now we adopt the similar arguments as in \cite{kushner:1997} and have
\begin{eqnarray}
E_{\xi11}(E_{\xi11}^+E_{\xi10}-E_{\xi_211}^+E_{\xi_210})=0,\label{proof3-4}\\
E_{\xi_111}(E_{\xi_111}^+E_{\xi_110}-E_{\xi_211}^+E_{\xi_210})=0,\label{proof3-5}
\end{eqnarray}
where $^+$ represents the Moore-Penrose generalized inverse. For symmetric $\xi_1$, we have $E_{\xi_111}=Q_{\xi_1}\otimes B_t/(t-1)$ with $1_t'C_{\xi ij}1_t=0$ for $i,j=1,2$. For an arbitrary matrix $A$, let $C(A)$ denote the column space of $A$ and $C^{\perp}(A)$ denote the orthogonal complement space. Note that (\ref{proof3-5}) indicates 
\begin{eqnarray*}
C(E_{\xi_111}^+E_{\xi_110}-E_{\xi_211}^+E_{\xi_210})\subset C^{\perp}(E_{\xi_111}).
\end{eqnarray*}
Meanwhile, we can verify that
\begin{eqnarray*}
C^{\perp}(E_{\xi_111})\subset C^{\perp}(E_{\xi_111}^+E_{\xi_110}-E_{\xi_211}^+E_{\xi_210}).
\end{eqnarray*}
Thus, we have $C(E_{\xi_111}^+E_{\xi_110}-E_{\xi_211}^+E_{\xi_210})\subset C^{\perp}(E_{\xi_111}^+E_{\xi_110}-E_{\xi_211}^+E_{\xi_210})$ which indicates $E_{\xi_111}^+E_{\xi_110}-E_{\xi_211}^+E_{\xi_210}=0$ and so
\begin{eqnarray}
E_{\xi_211}^+E_{\xi_210}&=&E_{\xi_111}^+E_{\xi_110}\notag\\
&=&Q_{\xi_1}^{-1}l_{\xi_1}\otimes B_t\label{proof3-6}\\
&=&-x^*\otimes B_t,\notag
\end{eqnarray}
in view of Theorem \ref{thm:1}($iii$). Now (\ref{proof3-2}) is derived from (\ref{proof3-4}) and (\ref{proof3-6}). Equations (\ref{proof3-1}) and (\ref{proof3-3}), and the sufficiency can all be derived analogously to the proof of Theorem 3 in \cite{zheng:2015}.

\vspace{.4cm}

{\noindent\bf Proof of Theorem \ref{thm:3.4}.} Let $\triangledown q_{s}(x)$ (resp. $\triangledown q_{\xi}(x)$) be the gradient of the bivariate function $q_{s}(\cdot)$ (resp. $q_{\xi}(\cdot)$) evaluated at point $x$. Note that (\ref{eqn-1}) is equivalent to $\sum_{s\in {\cal T}} p_s \triangledown q_s(x^*)=0$. Suppose (\ref{eqn-1}) and (\ref{eqn-2}) hold, then $q_{\xi}(x)$ reaches its minimum at $x^*$ and hence $y_{\xi}=q_{\xi}(x^*)=\sum_{s\in \cal T}p_s q_s(x^*)=\sum_{s\in \cal T}p_s y_*=y_*=y^*$. So $\xi$ is universally optimal due to the former conclusions and hence the sufficiency of the theorem. The necessity follows from the three former conclusions in view of $\triangledown q_{\xi}(x)=\sum_{s\in {\cal S}} p_s\triangledown q_s(x)$.

\vspace{.4cm}

In the rest of this appendix, discussions are all based on Models (\ref{model:2}) and (\ref{model:3}).
For the convenience of later analysis, we introduce the following notation under Model (\ref{model:2}). The covariance matrix $\Sigma$ is assumed to be type-H from now on.
\begin{eqnarray}\label{eqn:030801}
q_s(x)=k-\frac{\chi_s}{k}+4(\gamma_s-k)x+(6k+2\psi_s-8\gamma_s)x^2.
\end{eqnarray}

\vspace{.4cm}

{\noindent\bf Proof of Theorem \ref{thm:5}.} Theorem \ref{thm:5}($i$) can be derived similarly as Lemma 5 in \cite{zheng:2015}. Theorem \ref{thm:5}($ii$)--($iii$) can be proved similarly to Theorem 5 in \cite{zheng:2015}.

\vspace{.4cm}

It should be mentioned here that proofs of Theorems \ref{thm:table} and \ref{thm:121101} frequently use the concept of {\it type}-$j$ {\it treatment} given in the proof of Theorem \ref{thm:table} when $t=3$.
We shall break the proof of Theorem \ref{thm:table} down to three cases according to different $t$: ($i$) $t=2$, ($ii$) $t=3$ and ($iii$) $t\geq 4$ and $4\leq k\leq 10$.\\

{\noindent\bf Proof of Theorem \ref{thm:table}.} The proof is teared down to three parts according to ($i$) $t=2$, ($ii$) $t=3$ and ($iii$) $t\geq 4$ and $4\leq k\leq 10$.

\vspace{.2cm}

Consider $t=2$. Suppose $k=2\lambda$ for some integer $\lambda$. Consider the intersection of $q_{s_1}(x)$ and $q_{s_2}(x)$, i.e., $x_0={(\lambda-1)}/{(2\lambda-1)}=0.5-1/(4\lambda-2)$. Simple analysis reveals that to have a larger $\psi_s$ than $s_1$, the sequence must have the pattern of $(11...111212...1212)$. Obviously we have $\psi_s\leq 2\lambda-2$. Thus, to achieve $q_{s}(x_0)\geq q_{s_1}(x_0)$ for such a sequence $s$ we need $f_1-f_2\geq \lambda-1$. In this case we can show that the increase from the $\gamma_s$ is even less than the increase in $\chi_s/k$ and so we always have $q_{s}(x_0)<q_{s_1}(x_0)$. Similar analysis can be adapted to $k=2\lambda+1$ and is omitted here for simplicity.

\vspace{.2cm}

Consider $t=3$. Suppose $k\geq 48$. For the other cases of $k<48$, the results have been verified with the help of computer codes.
When $k=3u$ with a positive integer $u$, for $s_1$ and $s_2$, a simple calculation reveals that the intersection is achieved at $x_0=(2u-2)/(4u-3)$.
Note that we have at most $3$ different treatments. We can define $\gamma_{s,i}$ and $\psi_{s,i}$ as the contribution to $\gamma_s$ and $\psi_s$ from treatment $i=1,2,3$. For $s_1$, $\gamma_{s_1}=3u-3$ and $\psi_{s_1}=3u-6$. For $s_2$, $\gamma_{s_2}=u$ and $\psi_{s_2}=3u-4$.
If there shall be some sequence $s_0$ such that $q_{s_0}(x_0)\geq q_{s_1}(x_0)$, we must have all $3$ different treatments in $s_0$ since otherwise $\chi_s/k$ will be too large that any possible increase in $\gamma_s$ and $\psi_s$ can not compensate for the loss from the increasing of $\chi_s/k$ (for simplicity, we will simply say $\chi_s/k$ is {\it too large} for this phenomenon). Also, we have $\psi_{s_0}>\psi_{s_1}=3u-6$ since $\gamma_{s_1}$ has already achieved its maximum over all sequences which has at least $3$ treatments. Now we consider $\psi_{s_0}=3u,3u-1,3u-2,3u-3,3u-4,3u-5$. Note that $\psi_{s_0,i}=f_i$ if and only if $k$ is even and $s_0=(i,a_1,i,a_2,\ldots,i,a_{k/2})$ where $a_j\neq 1$ for $j=1,\ldots,k/2$. 

{\it It should be especially emphasized here that the following definition of {\it type-$j$} treatment is vital in proofs of the current theorem and Theorem \ref{thm:121101}.}
For any sequence $(a_1,\ldots,a_t)$, define $\psi_{s,i}=\#\{j:a_j=i,a_{j+2}=i\}$. We call treatment $i$ with $\psi_{s,i}=f_i-j$ as type-$j$ treatment and use $n_j$ to denote the number of type-$j$ treatments. It should be emphasized here that $n_0=0$ for all sequences in $\cal T$ when $t>2$ since otherwise $\chi_{s_0}/k$ will be too large (this conclusion will be discussed in detail in the proof of Theorem \ref{thm:121101}).

So it is impossible to have $\psi_{s_0}=3u,3u-1,3u-2$ since we need at least one type-$0$ treatment. To achieve $\psi_{s_0}=3u-3$, we have $n_0=0$, $n_1=3$. This can be achieved if and only if $s_0$ is in the same equivalence class as $(1212...1212111...111313...1313)$ such that we have too many $1$ ($f_1>k/2$) such that $\chi_s/k$ is also too large. Now we come to the case of $\psi_{s_0}=3u-4$. Then we have $n_0=0$, $n_1=2$, $n_2=1$ since there are at least $3$ different treatments. When $\psi_{s_0,1}=f_1-1$, it is either $f_{s_1}\geq k/2$ or $f_{s_1}< k/2$ and this treatment $1$ appears in the form of $s=(1a_11a_21...1a_{f_1}a_{f_1+1}...)$ with $a_j\neq 1$. Thus, $\gamma_{s_0,1}=0$. Similarly, we know $\psi_{s_0,2}=f_2-1$ indicates $\gamma_{s_0,2}=0$ and $\psi_{s_0,3}=f_3-2$ indicates $\gamma_{s_0,3}\leq f_3-1$. Then, compared with $s_2$, it is only possible to win from the increase of $\gamma_s$, but this benefit comes with a larger increase of $\chi_s/k$ such that we still have $q_{s_2}(x_0)>q_{s}(x_0)$. When $\psi_{s_0}=3u-5$, compared with $s_2$, $\gamma_s$ can be increased by at most $2k/3$ while the loss in $\psi_s$ overwhelms the benefit from the increase of $\gamma_s$. Thus, for $k=3u$, we have proved the corresponding result. Similar analysis also holds for $k=3u+1$ and $k=3u+2$ and so is omitted here.

\vspace{.2cm}

Consider $t=4$ and $4\leq k\leq 10$.
For $k=3$ and $s_1=(123)$, by (\ref{eqn:030801}) we have
\[{q}_{s_1}(x)=2-\frac{32}{3}x+\left(16-\frac{2(t+1)}{3t}\right)x^2,\]
The minimizer ${q}_{s_1}(x)$ is $x^*=\frac{8t}{23t-1}$ with the minimum value $y^*={q}_{s_1}(x^*)=\frac{2}{3}\frac{5t-3}{23t-1}$. Note $s=(aaa)$, $(abb)$, $(bab)$ or $s_1$ in this case. It's easy to verify $g_s<g_{s_1}$ if $s\neq s_1$.
Besides, ${q}_{s_1}^{\prime}(x^*)=0$, hence ${\cal T}_0=\langle(123)\rangle$ by Lemma 3.2 in \cite{kushner:1997}.

Similar tedious analysis can be applied to $4\leq k\leq 8$. However we may adopt a much simpler analysis with the help of computer codes. For $4\leq t\leq k$, we can search for the optimal sequence by Algorithm \ref{alg:T}. For $t>k$, we can prove that $p_{s_1}'(0.5)\leq 0$ and $p_{s_2}'(0.4)\geq 0$ for all $s_1$ and $s_2$ listed in Table \ref{tb:theorem}. For the former $x^*$ and the expression of $q_s(x^*)$, we can easily conclude that the maximum value of $q_{s}(x^*)$ is still achieved by the former ${\cal T}$ when $t\leq k$ since the change of $t$ does not change the rank of $q_s(x^*)$ with respect to $s$. Since $p_{s_1}'(0.5)\leq 0$ and $p_{s_2}'(0.4)\geq 0$ always hold, we can see $p_{s_1}'(x^*)\leq 0$ and $p_{s_2}'(x^*)\geq 0$. Thus, $x^*$ is still the minimizer of $\max_sq_s(x)$. The conclusion right follows.

\vspace{.4cm}

{\noindent\bf Proof of Theorem \ref{thm:11}.} Here we only discuss the case of $k>20$. {The other cases with $k\leq 20$ have been verified by computer codes.} We shall begin with the proof of  $x^*\in [0.4,0.5]$.
By direct calculations we have the following results for the $q_s(0.4)$, $q'_{s}(0.4)$, $q_s(0.5)$, $q'_{s}(0.5)$. 
\begin{eqnarray}
&&q_s(0.4)=0.36k-k^{-1}{\chi_s}+0.32\gamma_s+0.32\psi_s,\label{eqn:030601}\\
&&q'_{s}(0.4)=0.8k+1.6\psi_{s}-2.4\gamma_{s},\label{eqn:de4}\\
&&q_s(0.5)=0.5k-k^{-1}{\chi_s}+0.5\psi_s,\label{eqn:func5}\\
&&q'_{s}(0.5)=2k+2\psi_{s}-4\gamma_{s}.\label{eqn:de5}
\end{eqnarray}
Consider all $f_1$ elements $1$ in the sequence, if $f_1\neq 0$. By replacing all these elements to the first $f_1$ position(s) in the new sequence, we can see that the value of $0.32\gamma_s+0.32\psi_s$ increased by these elements is increased. Meanwhile, the $\chi_s$ is not changed since we are only re-ordering the elements in this sequence. This argument applies to all treatments in the original sequence $s$. Thus, one of the sequences of the form $(1\ldots12\ldots23\ldots3\cdots)$ maximizes $q_s(0.4)$. Also, it is easy to show that $q_s'(0.4)\leq 0$ for such sequences and so $x^*\geq 0.4$. Similar analysis reveals that $q'_s(0.5)\geq 0$ for all $s\in\cal S$ and so $x^*\leq 0.5$. The former analysis reveals $x^*\in[0.4,0.5]$.

The following analysis will address the lower bound of the design supported on one fixed equivalence class of sequences (and its due equivalence class). Here we give the outline of the proof. First, we find a sequence that maximizes $q_s(0.4)$. This sequence, denoted by $s_{i_0}$ will be shown to have $q_{s_{i_0}}(0.4)\leq 0$, $q_{s_{i_0}}(0.5)\geq 0$. Since it maximizes $q_s(0.4)$ and $x^*\in[0.4,0.5]$, we have $y^*\leq q_s(0.4)$. Then, the lower bound for the optimal design supported on $\langle s_{i_0}\rangle$ is at least $y^*_{i_0}/q_s(0.4)$ where $y^*_{i_0}=\min_xq_{s_{i_0}}(x)$ and this value is taken as the lower bound given in this theorem. Following this outline, the detailed analysis is given as follows.

Let $s_{i,f_1,\ldots,f_i}=(1_{f_1}',2\cdot 1_{f_2}',\cdots,i\cdot1_{f_i}')$. Former analysis has reveals that $q_s(0.4)$ ($s\in\cal S$) can be maximized by one of some $s_{i,f_1,\ldots,f_i}$. Now we find one maximizer and characterize the corresponding parameters $i_0,f_1,\ldots,f_{i_0}$. First we shall show $|\{j:1\leq j\leq i_0,f_j=1\}|\leq 1$ since otherwise we can merge two different treatments used for only one period to one treatment used in two neighbored periods. The increase in $\gamma_s(0.4)$ provides an increase of $0.32$ in $q_s(0.4)$ and the increase from $\chi_s/k$ (which decreases $q_s(0.4)$) is $2/k<0.32$k, which leads to contradiction. 
Now we will show that $|\{j:1\leq j\leq i_0,f_j=1\}|=0$. We can always increase $f_{j_{min}}$ ($f_{j_{min}}$ minimizes $f_1,\ldots,f_{i_0}$) by one and decrease $f_{j_{max}}$ ($f_{j_{max}}$ minimizes $f_1,\ldots,f_{i_0}$) by one such that $\gamma_s$ and $\psi_s$ are not changed while $\chi_s/k$ is decrease and $q_s(0.4)$ is increased. Thus, we know $q_s(0.4)$ can be maximized by either of the following two types of sequences: ($i$) $s_{i,1}=(1\cdot1'_{f_1},2\cdot1'_{f_2},\ldots,i\cdot1'_{f_i},i+1)$ with $1+f_i\geq f_1\geq f_2\geq\cdots\geq f_i\geq 2$, and ($ii$) $s_{i,2}=(1\cdot1'_{f_1},2\cdot1'_{f_2},\ldots,i\cdot1'_{f_i})$ with $1+f_i\geq f_1\geq f_2\geq\cdots\geq f_i\geq 2$.
By increasing $i$ by one and rebalancing $f_1,\ldots,f_i$, we can see that the decrease in $\chi_s/k$ is at least $(k-1)/i-(k-1)/(i-1)-(i+1)/k$ for case ($i$) and $k/i-k/(i-1)-(i+1)/k$ for case ($ii$). For $i=2,3$, we can show that this decrease is larger than $0.64$ with $k>20$, which overwhelms the decrease in $0.32\gamma_s+0.32\psi_s$, i.e., $0.64$. Thus, we have $i\geq 4$. So if we merge the one-period treatment into $f_{j_{min}}$ to increase it by one, the increase in $0.32\gamma_s+0.32\psi_s$ is $0.64$, while the increase in $\chi_s/k$ is less than $[(k/4+1)^2-(k/4)^2-1]/k=1/2<0.64$. Thus, we can do this merge and increase $q_s(0.4)$. So, we have proved $|\{j:1\leq j\leq i_0,f_j=1\}|=0$.

Now, we have proved that the maximizer of $q_s(0.4)$ can be found among case ($ii$), i.e., $s_{i,2}=(1\cdot1'_{f_1},2\cdot1'_{f_2},\ldots,i\cdot1'_{f_i})$ with $1+f_i\geq f_1\geq f_2\geq\cdots\geq f_i\geq 2$. For simplicity, we will simply write $s_{i,2}$ as $s_i$ and find the sequence we need from this type of sequences. So we need to decide the value of $i$. Simple calculation reveals
\begin{eqnarray}\label{eqn:052601}
q_{s_i}(0.4)=0.36k-\chi_{s_i}/k+0.32(k-i)+0.32(k-2i)=k-(0.96i+\chi_{s_i}/k),
\end{eqnarray}
where $\chi_{s_i}=(k-i\lfloor k/i\rfloor)(\lfloor k/i\rfloor+1)^2+(i-k+i\lfloor k/i\rfloor)(\lfloor k/i\rfloor)^2$. Note that
\begin{eqnarray}\label{eqn:120601}
\chi_{s_i}&=&(k-i\lfloor k/i\rfloor)(\lfloor k/i\rfloor+1)^2+(i-k+i\lfloor k/i\rfloor)(\lfloor k/i\rfloor)^2\\
&=&(k-i\lfloor k/i\rfloor)(k/i+\Delta_1)^2+(i-k+i\lfloor k/i\rfloor)(k/i-\Delta_2)^2,\notag
\end{eqnarray}
where $\Delta_1+\Delta_2=1$.
Obviously, $(k-i\lfloor k/i\rfloor)\Delta_1=(i-k+i\lfloor k/i\rfloor)\Delta_2$, and $i\Delta_1=(i-k+i\lfloor k/i\rfloor)$. Direct analysis reveals
\begin{eqnarray}\label{eqn:120602}
\chi_{s_i}&=&k^2/i+i(1-\Delta_1)\Delta_1^2+i(1-\Delta_1)^2\Delta_1=k^2/i+i(1-\Delta_1)\Delta_1\leq k^2/i+i/4.
\end{eqnarray}

We can see that
\begin{eqnarray*}
0.96i+k/i\leq (0.96i+\chi_{s_i}/k)\leq 0.96i+k/i+i/4k.
\end{eqnarray*}
The lower bound above is minimized at $i=\sqrt{k/0.96}$ and the upper bound above is minimized at $\sqrt{k/(0.96+1/4k)}$. Thus we claim that the optimal $i_0$ locates in $(\sqrt{k/0.96}-1.5,\sqrt{k/0.96}+1.5)$ using the following analysis. Take the largest integer $i_1=\sqrt{k/0.96}+\delta\in[\sqrt{k/0.96},\sqrt{k/0.96}+1]$.
For arbitrary $i_2=\sqrt{k/0.96}+B\geq \sqrt{k/0.96}+1.5$, we have
\begin{eqnarray*}
&&q_{s_{i_1}}(0.4)-q_{s_{i_2}}(0.4)\notag\\
\geq&&0.96(\sqrt{k/0.96}+B)+k/(\sqrt{k/0.96}+B)-0.96(\sqrt{k/0.96}+\delta)\notag\\
&&-k/(\sqrt{k/0.96}+\delta)-(\sqrt{k/0.96}+\delta)/k\notag\\
=&&0.96(B-\delta)-(B-\delta)k/[(\sqrt{k/0.96}+B)(\sqrt{k/0.96}+\delta)]-(\sqrt{k/0.96}+\delta)/k\notag\\
=&&0.96(B-\delta)[(B+\delta)\sqrt{k/0.96}+B\delta]/[(\sqrt{k/0.96}+B)(\sqrt{k/0.96}+\delta)]/k\notag\\
&&-(\sqrt{k/0.96}+\delta)~~~~~~~~~~>~0 ~(~{\rm for} ~k\geq 20),
\end{eqnarray*}
and so $i_0<\sqrt{k/0.96}+1.5$. Similarly, we can show $i_0>\sqrt{k/0.96}-1.5$.

Further more, for $i_0\in(\sqrt{k/0.96}-1.5,\sqrt{k/0.96}+1.5)$ and $i(t,k)=\min(t,i_0)$, we have
\begin{eqnarray}\label{eqn:lowerbound}
{\rm effi}(\langle s_{i}\rangle)&=&\frac{\min_xq_{s_{i}}(x)}{\max_{s\in\cal S}q_s(0.4)}=\frac{k-\chi_{s_{i}}/k-i}{k-\chi_{s_{i}}/k-0.96i}=1-\frac{0.04i}{k-\chi_{s_{i}}/k-0.96i}\\
&\geq&1-\frac{0.04i}{k-k/i-i/k-0.96i}\label{eqn:uuuuu}.\\
{\rm effi}(\langle s_{i(t,k)}\rangle)&>&1-\frac{0.04\sqrt{k/0.96}+0.06}{k-2.5-0.96\sqrt{k/0.96}}.\label{eqn:uuuuu2}
\end{eqnarray}
The last inequality (\ref{eqn:uuuuu2}) comes from the fact that (\ref{eqn:uuuuu}) decreases in $i$.
This gives the efficiency lower bound. Note that for each $s_i$, we have $q_{s_i}'(0.5)=0$. Now we search for the $i$ which maximizes ${\rm effi}(\langle s_{i}\rangle)$ defined in (\ref{eqn:lowerbound}), denoted by $i^*=\min\{{\arg\max}_{i} ~q_{s_i}(0.5),t\}$.

Obviously, $i^*$ maximizes $k-\chi_{s_{i}}/k-i$. Then we need to find the $i$ that minimizes $\chi_{s_{i}}/k+i$, where
\begin{eqnarray}
\chi_{s_{i}}/k+i=(i+k/i)+i(1-\Delta_1)\Delta_1/k.\label{proof3-7}
\end{eqnarray}
If $\sqrt{k}\in\mathbb{Z}$, then $i$ minimizes $(i+k/i)$. When we decrease $i$ by $1$, then the increase in $(i+k/i)$ equals to $1/(\sqrt{k}-1)$ while the decrease of $i(1-\Delta_1)\Delta_1/k$ is at most $1/(4\sqrt{k})$. When we increase $i$, then both two terms in (\ref{proof3-7}) increase. Thus, in this case, ${\arg\max}_{i} ~q_{s_i}(0.5)=\sqrt{k}$. If $\sqrt{k}\not\in\mathbb{Z}$, consider $\lfloor\sqrt{k}\rfloor$ and $\lceil\sqrt{k}\rceil$. When we consider $i<\lfloor\sqrt{k}\rfloor$ or $i>\lceil\sqrt{k}\rceil$, we can show that $\chi_{s_{i}}/k+i$ increases using the same analysis. Thus, ${\arg\max}_{i} ~q_{s_i}(0.5)\in\{\lfloor\sqrt{k}\rfloor,\lceil\sqrt{k}\rceil\}$ in this case.

\begin{lemma}\label{lem:b1}
Suppose $\Sigma$ is of type-H and $k>10$, $t>3$. For any sequence $s$ which contains at least one type-$0$ treatment and arbitrary optimal approximate design $\xi$, we have $\xi(s)=0$, which means $s\not\in\cal T$.
\end{lemma}
\begin{proof} There are two ways to include a type-$0$ treatment: ($i$) $(1,\ldots,1)$, ($ii$) $(1,a_1,1,a_2,$ $\ldots,1,a_{k/2})$ with an even $k$ and $a_1,\ldots,a_{k/2}\neq 1$. The sequence of pattern ($i$) can be easily excluded from $\cal T$ since otherwise $y^*=0$. Now we focus on ($ii$). Suppose a sequence, say $s$, follows pattern ($ii$) and has $r(>0)$ treatments other than 1. Simple calculation reveals $\gamma_s=0$ and $\psi_s\leq k-r$.

The result of this lemma holds if we can show that, for any sequence $s_0$ of pattern ($ii$), there exists a sequence $s$ in ${\cal S}^*$ and not of pattern ($ii$) such that $q_s(x)>q_{s_0}(x)$ for all $x\in[0.4,0.5]$. Now we will dive into a direct but complex analysis to claim this hypothesis. Before doing this, for convenience, we rewrite $q_s(x)$, $q_s'(x)$ and $q_s(0.5)$ as follows.
\begin{eqnarray}
q_s(x)&=&k-\frac{\chi_s}{k}+4(\gamma_s-k)x+(6k+2\psi_s-8\gamma_s)x^2;\label{eqn:120603}\\
q_s'(x)&=&4(\gamma_s-k)+2x\cdot (6k+2\psi_s-8\gamma_s);\label{eqn:120604}\\
q_s(0.5)&=&0.5k-\frac{\chi_s}{k}+0.5\psi_s.\label{eqn:120605}
\end{eqnarray}

If $s_0$ has $\psi_s< k-r$, we can resort $a_1,\ldots,a_{k/2}$ in the descending order such that we can always get a new sequence $s'$ of pattern ($ii$) such that $\psi_{s'}= k-r$. From (\ref{eqn:120603}), we can see that nothing is changed except $\psi$ is increased by changing $s_0$ to $s'$. Thus, we have $q_{s'}(x)>q_{s_0}(x)$ for arbitrary $x\in[0.4,0.5]$. So, we only need to show there exists a sequence $s$ in ${\cal S}^*$ and not of pattern ($ii$) such that $q_s(x)>q_{s'}(x)$ for all $x\in[0.4,0.5]$.

From (\ref{eqn:120604}), we have $q_{s}(x)-q_{s'}(x)=(4-16x)(\gamma_{s}-\gamma_{s'})+4x(\psi_s-\psi_{s'})$. If the number of different treatments in $s'$ is less or equal to that in $s\in{\cal S}^*$, we have $\psi_s-\psi_{s'}<0$. And also, we have $\gamma_{s'}=0$ and $4-16x<0$ for $x\in[0.4,0.5]$. Thus, we conclude that $q_{s}(x)-q_{s'}(x)<0$ if the number of different treatments in $s'$ is less or equal to that in $s\in{\cal R}_2(k,t)$. If we can find such an $s\in{\cal S}^*$ which simultaneously satisfies $q_{s}(0.5)-q_{s'}(0.5)>0$, our claim is verified by this $s$. Now we move on to address this $s$.

If $s'$ has two different treatments, it means $s_0\in\langle(1,2,1,2,\ldots,1,2)\rangle$.
Then, $q_{s'}(0.5)=k-k/2$. Let $s$ be the sequence in ${\cal R}_2(k,t)$ with $t_1=0$ and $r=2$. We have $q_s(0.5)\geq k-k/4-1/k-2$. Simple calculation reveals $q_{s}(0.5)-q_{s'}(0.5)>0$ when $k\geq 9$.

If $s'$ has three different treatments, $q_{s'}(0.5)\leq k-3k/8-1$. Let $s$ be the sequence in ${\cal S}^*$ with $t_1=0$ and $r=2$. We have $q_s(0.5)\geq k-k/4-1/k-2$. Simple calculation reveals $q_{s}(0.5)-q_{s'}(0.5)>0$ when $k\geq 9$, still.

If $s'$ has $z+1$ ($z\geq 3$) different treatments, $q_{s'}(0.5)\leq k-k/4-k/4z-0.5z$. If $z$ is odd, let $s$ be the sequence in ${\cal S}^*$ with $t_1=0$ and $r=(1+z)/2$. We have $q_s(0.5)\geq k-k/(1+z)-(1+z)/4k-0.5(1+z)$. Simple calculation reveals $q_{s}(0.5)-q_{s'}(0.5)>0$ when $k\geq 9$, still. If $z$ is even, let $s$ be the sequence in ${\cal S}^*$ such that $t_1=1$, $r=z/2$ and the difference between frequencies of treatments is no larger than $1$. We have $q_s(0.5)\geq k-k/(1+z)-(1+z)/4k-0.5(2+z)$. Simple calculation reveals $q_{s}(0.5)-q_{s'}(0.5)>0$ when $k\geq 15$.

Thus, for $k\geq 15$, we have proved our hypothesis. For $10<t\leq 14$, the result of this lemma can be proved by ergodic searching codes whose computational cost is still affordable for such a small $k$.
\end{proof}

\begin{lemma}\label{lem:120701}
Suppose $\Sigma$ is of type-H and $k>10$, $t>3$.  Then $x^*\in[0.4,0.5)$ where $x^*$ is defined in (\ref{eqn:2262}).
\end{lemma}
\begin{proof}
First, consider the $q_s(0.4)$ in (\ref{eqn:030601}). Obviously, $s=(1,1,1,\ldots,1)$ has the largest $0.32(\gamma_s+\psi_s)$, but the $\chi_s$ is too large. For arbitrary $s$, we can rearrange its treatments in ascending order without decreasing $0.32(\gamma_s+\psi_s)$. Since $\chi_s$ is unchanged in this rearrangement, we have $q_{s'}(0.4)\geq q_s(0.4)$ where $s'$ is the rearranged version of $s$ according to the analysis above. This new $s'$ is of pattern $(1,1,1,2,2,2,2,3,3,3,\ldots)$. 

For the $s'$ generated above, consider treatments with frequency $1$, i.e., treatments appear only once. If there are more than one such treatments, we can choose two such treatments and change one of them to the other. For example, if treatments $1$ and $2$ both appear only once, we change treatment $2$ to treatment $1$ and put them in neighboring periods. In this process, $\gamma_{s'}$ is increased by $1$ and $\chi_{s'}/k$ is increased by $(2^2-1-1)/k=2/k<0.32$ when $k>6$. Thus, we can always generate a new sequence $s''$ which has at most $1$ treatment which frequency $1$ such that $q_{s''}(0.4)\geq q_{s'}(0.4)$. Now, for other treatments other than the frequency-$1$ treatment, we can balance their frequencies such that $|f_{i_1}-f_{i_2}|\leq 1$ without decreasing $0.32(\gamma_{s''}+\psi_{s''})$ but strictly decreasing $\chi_{s''}$. For example, if treatment $2$ appears in $3$ periods and treatment $3$ appears in $5$ periods, we can change the treatment $3$ in one period to treatment $2$ and rearrange the sequence again. Then both $\gamma_{s''}$ and $\psi_{s''}$ are unchanged in this process while $\chi_{s''}$ is decreased by $(5^2+3^2-2\cdot 4^2)/k=2/k$. Thus, we can always find one sequence which has pattern ($i$) $(1,2\cdot 1_{f_2}',3\cdot 1_{f_3}',\ldots,i\cdot 1_{f_i}')$ with $|f_{i_1}-f_{i_2}|\leq 1$ for all $2\leq i_1< i_2\leq i$, or pattern ($ii$) $(1\cdot 1_{f_1}',2\cdot 1_{f_2}',3\cdot 1_{f_3}',\ldots,i\cdot 1_{f_i}')$ with $|f_{i_1}-f_{i_2}|\leq 1$ for all $1\leq i_1< i_2\leq i$. For both two types, we can show that $q_{s''}'(0.4)\leq 0$. Thus, $x^*\geq 0.4$.

Consider $q_{s}(0.5)=0.5k-\chi_s/k+0.5\psi$. For any such sequence, rearrange treatments in ascending order as in the former analysis for $q_s(0.4)$. Let $s'$ denote the new sequence which has pattern $(1,2,2,2,3,3,3,3,\ldots)$. Let $n^{(1)}$ denote the number of treatments with frequency $1$ and $n^{(2)}$ denote the number of treatments with frequencies larger than $1$. Then we have $q_{s'}'(0.5)=2n^{(1)}\geq 0$. If $n^{(2)}=0$, then $q_{s'}'(0.5)=2k>0$ and $s'=(1,2,3,4,5,6,\ldots,k)$. We change treatment $3$ to $1$ to generate $s'=(1,2,1,4,5,6,\ldots,k)$. And it is obvious that $q_{s''}(0.5)-q_{s'}(0.5)=0.5-2/k>0$ and $q_{s''}(0.5)=2k+2>0$. If $n^{(2)}>0$, we can pick one type-$v$ ($v>1$) treatment, say treatment $2$ in $s'=(1,2,2,2,3,3,3,3,\ldots)$ with an arbitrary treatment other than $2$ (we can always find this treatment since we have proved the extremely poor performance of sequence with only one treatment). Suppose it is treatment $1$ with $f_1$ as its frequency. Then, we can balance the frequencies of treatments $1$ and $2$ and then write the combination of these two treatments as $(1,2,1,2,1,\ldots,2,1)$ or $(1,2,1,2,1,\ldots,1,2)$ or $(2,1,2,1,\ldots,2,1,2)$. By doing so, we see the $\chi_s$ is not increased while there is one more type-$1$ treatment with all other treatments of type-$2$. Thus, $\psi_s$ is increased by $1$ at least in this process. Let $s''$ denote this new sequence. We can see $q_{s''}(0.5)>q_{s'}(0.5)$. Note that $\gamma_{s'}\geq\gamma_{s''}$ and $\psi_{s'}<\psi_{s''}$. We know $q_{s''}'(0.5)-q_{s'}'(0.5)=2(\psi_{s''}-\psi_{s'})-4(\gamma_{s''}-\gamma_{s'})>0$. Note that we have already shown that $q_{s'}'(0.5)\geq 0$. Thus, $q_{s''}'(0.5)>0$.

The analysis above shows us two outlets for an arbitrary sequence $s$: ($i$) If $q_s(0.5)>q_{s'}(0.5)$, it means $\psi_{s}>\psi_{s'}$ since they have the same $\chi$ value in $q_{s}(0.5)=0.5k-\chi_s/k+0.5\psi$. Note that $s'$ gives maximum $\gamma$ value among all possible rearrangements, we have $\gamma_s\leq \gamma_{s'}$ and so $q_{s}'(0.5)-q_{s'}(0.5)=2(\psi_{s}-\psi_{s'})-4(\gamma_{s}-\gamma_{s'})>0$. ($ii$) If $q_s(0.5)\leq q_{s'}(0.5)$, we have found an $s''$ in the last paragraph such that $q_{s''}(0.5)>q_{s'}(0.5)\geq q_s(0.5)$ and $q_{s''}'(0.5)>0$. It means either it has positive derivative at $0.5$ or there exist a sequence with a larger value at $0.5$ and positive derivative. Thus, $x^*<0.5$.
\end{proof}

Suppose a period is assigned with a non-type-$0$ and non-type-$1$ treatment, and its left and right neighbors are assigned with the same type-$1$ treatment, we call this period an {\it isolated period}. The number of all isolated periods in a sequence $s$ is denoted by $ip(s)$.

\begin{lemma}\label{lem:123101}
For any subset of $\{a_1,\ldots, a_{t'}\}\subset \{1,\ldots,t\}$ be a subset with $t'(\leq t)$ different treatments. Let $f_{a_1},\ldots,f_{a_{t'}}$ denote $t'$ positive numbers. Consider the set $\{s\in{\cal S}:a_1,\ldots, a_{t'}~{\rm are~ type-1 ~ treatments~with~frequencies~}f_{a_1},\ldots,f_{a_{t'}}\}$ denoted by ${\cal B}$. Define
\begin{eqnarray*}
\Delta=\min\bigg\{\sum_{j=1}^{t''}f_{a_{i_j}}-\sum_{j=t''+1}^{t'}f_{a_{i_j}}:~\sum_{j=1}^{t''}f_{a_{i_j}}-\sum_{j=t''+1}^{t'}f_{a_{i_j}}\geq 0, ~0\leq t''\leq t',\\
(i_1,\ldots,i_{t'}) {\rm~is ~a ~permutation ~of} ~(1,\ldots,{t'})\bigg\},
\end{eqnarray*}
with $\sum_{i+1}^i$ defined as the summation of no item and so always equals $0$. Then we have $\min_{s\in{\cal B}} ip(s)=\max\{0,\Delta-1\}$.
\end{lemma}
\begin{proof}
We prove this lemma with an illustrative example. Let $s_0=(1,1,2,3,2,3,2,3,3,$ $4,3,4,3,4,5,5,6,7,6,7)$. Treatment $3$ is type-$2$, so the $4$th ,$6$th, $11$th, and $13$th periods are all isolated periods. It should be mentioned that the $18$th period is non-isolated since it is assigned with a type-$1$ treatment $7$.

Remove all periods in $s$ which is neither isolated nor assigned with type-$1$ treatments, and the resulting sequence is denoted by $s'$. 

In our illustrative example $s_0$, we show this process as follows.
The first two periods and last two periods in $s_0$ are assigned with type-$2$ treatments $1$ and $5$ and also non-isolated and are so removed. Periods in $s_0$ assigned with $2$, $4$, $6$, or $7$ (all of type-$1$), are maintained. The $4$th ,$6$th, $11$th, and $13$th periods are all isolated periods, and are so maintained. Here we use $\#$ to represent removed periods. The $8$th and $9$th periods are assigned with type-$2$ treatment $2$ and are also non-isolated, and so removed. And then, the original sequence $s_0$ becomes $(\#,\#,2,3,2,3,2,\#,\#,4,3,4,3,4,\#,\#,6,7,6,7)$.
Remove all empty periods labeled in $\#$ and finally we have $s_0'=(2,3,2,3,2,4,3,4,3,4,6,7,6,7)$.

Now consider the new sequence $s'$ which contains all periods assigned with type-$1$ treatments. Suppose the length of $s'$ ($s_0'$ in our illustrative example) is $n_0$, then the number of isolated periods in $s$ ($s_0$ in our illustrative example) equals to $n_0-\sum_{i=1}^{t'}{f_{a_i}}$ Let ${\cal A}_{\rm odd}$ and ${\cal A}_{\rm even}$ denote all type-$1$ treatments in odd and even periods of $s'$, respectively. Obviously, one type-$1$ treatment can not appear in both ${\cal A}_{\rm odd}$ and ${\cal A}_{\rm even}$.
Thus, we have $n_1+\sum_{a\in{\cal A}_{\rm odd}}f_a$ odd periods and $n_2+\sum_{a\in{\cal A}_{\rm even}}f_a$ even periods, where $n_1$ and $n_2$ are the number of isolated periods in odd and even periods, respectively. Take $s_0'$ as example, we have ${\cal A}_{\rm odd}=\{2,6\}$, ${\cal A}_{\rm even}=\{4,7\}$, $n_1=2$ and $n_2=2$. Note that, when $n_0$ is even, $n_1+\sum_{a\in{\cal A}_{\rm odd}}f_a=n_2+\sum_{a\in{\cal A}_{\rm even}}f_a$, and when $n_0$ is odd, $n_1+\sum_{a\in{\cal A}_{\rm odd}}f_a=1+n_2+\sum_{a\in{\cal A}_{\rm even}}f_a$. Simple calculation shows that the minimum value of $n_1+n_2$ equals to $\max\{\Delta,\Delta-1\}$, which completes the proof. The arrangement of these type-$1$ treatments can be found at the beginning of the proof of Lemma \ref{lem:010201} and is so omitted here.
\end{proof}

Given a sequence $s$, let ${\cal A}(s)$ and ${\cal A}_i(s)$ denote the set of all different treatments and all type-$i$ treatments assigned to $s$, correspondingly. Without special declaration, elements in ${\cal A}(s)$ and ${\cal A}_i(s)$ are all arranged in ascending order. Let ${\cal F}(s)$ and ${\cal F}_i(s)$ denote the corresponding frequencies of treatments in ${\cal A}(s)$ and ${\cal A}_i(s)$, if the set of treatments is not empty. For example, given $s=(1,2,1,2,3,3,3,4,4,4)$, then ${\cal A}(s)=\{1,2,3,4\}$, ${\cal F}(s)=(2,2,3,3)$; ${\cal A}_1(s)=\{1,2\}$, ${\cal F}_1(s)=(2,2)$; ${\cal A}_2(s)=\{3,4\}$, ${\cal F}_2(s)=(3,3)$; ${\cal A}_i(s)=\emptyset$, $i\geq 3$;
\begin{lemma}\label{lem:010201}
For an arbitrary sequence $s$, if $\min_{a\in{\cal A}(s)\setminus{\cal A}_1(s)}f_a\leq ip(s)+1$, one can find $s'$ such that $\min_{a\in{\cal A}(s')\setminus{\cal A}_1(s')}f_a> ip(s')+1$ and $q_s(x)<q_{s'}(x)$ for all $x\in[0.4,0.5)$.
\end{lemma}
\begin{proof}
Here, we adopt similar notations in the proof of Lemma \ref{lem:123101}. Define
\begin{eqnarray*}
(a'_1,\ldots,a'_{|{\cal A}_1(s)|},t''_0)=\arg_{(a_1,\ldots,a_{|{\cal A}_1(s)|},t''_0)}\min\bigg\{\sum_{j=1}^{t''}f_{a_{j}}-\sum_{j=t''+1}^{|{\cal A}_1(s)|}f_{a_{j}}:\\~\sum_{j=1}^{t''}f_{a_{j}}-\sum_{j=t''+1}^{|{\cal A}_1(s)|}f_{a_{j}}\geq 0, ~0\leq t''\leq |{\cal A}_1(s)|,\\
~(a_1,\ldots,a_{|{\cal A}_1(s)|}) {\rm~is ~a ~permutation ~of~element~in} ~{\cal A}_1(s)\bigg\}.
\end{eqnarray*}
In periods $1,3,5\ldots,2\sum_{j=1}^{t''_0}f_{a'_{j}}-1$, we arrange treatments $a_1',\ldots,a'_{t_0''}$ sequentially. In periods $2,4,6,\ldots,2\sum_{j=t''_0+1}^{|{\cal A}_1(s)|}f_{a'_{j}}$, we arrange the rest treatments. Then, the number of isolated periods is minimized. Now we sort the rest treatments (non-type-$1$ treatments) in by their frequencies such that the treatment on the left always has smaller or the same frequency as the treatment on the right, and fill them into all empty periods. The resulting sequence is denoted by $s'$. We give an example to illustrate this process. Let $s_0=(1,2,1,3,1,2,1,3,3,3,4,5,4,4,6)$. Then $1,5,6$ are type-$1$ treatments $2,3,4$ are non-type-$1$ treatments. The result sequence is $s'_0=(1,5,1,6,1,2,1,2,4,4,4,3,3,3,3)$.

Now we take a look at the resulting sequence $s'$. Let 
\begin{eqnarray*}
lp(s')=\max\bigg\{2\sum_{j=t''_0+1}^{|{\cal A}_1(s)|}f'_{a_{j}}~,~2\sum_{j=1}^{t''_0}f_{a'_{j}}-1\bigg\}.
\end{eqnarray*}
Consider the subsequence of $s'$ from period $lp(s')+1$ to the last period. Obviously, $\gamma_{s'}\geq \gamma_s$ and $\psi_{s'}\geq \psi_s$ since elements in $s'$ are either type-$1$ or type-$2$. If $\min_{a\in{\cal A}(s)\setminus{\cal A}_1(s)}f_a\leq ip(s)+1$, $s'$ has at least one more type-$1$ treatment than $s$ and so $\psi_{s'}>\psi_s$. Note that $\chi_s=\chi_{s'}$, we have $q_{s'}(x)>q_s(x)$ for arbitrary $x\in[0.4,0.5)$ according to (\ref{eqn:030801}). It should be emphasized that the sequence $s'$ generated here is called a {\it sorted competing sequence}. 
\end{proof}

For an arbitrary sorted competing sequence $s'$ generated in the proof of Lemma \ref{lem:010201}, we can further construct a dominating sequence $s''$ such that $q_{s''}(x)> q_{s'}(x)$ for arbitrary $x\in[0.4,0.5)$ as in the following Lemma \ref{lem:010202} if $s'$ has at least one isolated period.

\begin{lemma}\label{lem:010202}
For an arbitrary sorted competing sequence $s'$, if there is at least one isolated period in $s'$, one can find $s''$ such that $q_{s''}(x)> q_{s'}(x)$ for arbitrary $x\in[0.4,0.5)$.
\end{lemma}
\begin{proof}
In the our construction of $s'$, some non-type-$1$ treatments become type-$1$. For example, in $s_0'=(1,5,1,6,1,2,1,2,4,4,4,3,3,3,3)$ constructed in the proof of Lemma \ref{lem:010201}, the original non-type-$1$ treatment $2$ becomes type-$1$ and there is no isolated period. When there exist isolated periods, it is obvious that these periods in $s'$ must be assigned with the same treatment. Note that $\sum_{j=1}^{t''_0}f_{a'_{j}}\geq \sum_{j=t''_0+1}^{|{\cal A}_1(s)|}f_{a'_{j}}$, the first isolated period must appear be an even period. The corresponding treatment is denoted by $a_{\rm iso}$.

Let ${\cal A}_{\rm odd}$ and ${\cal A}_{\rm even}$ denote the treatments appearing in odd and even periods, respectively.

When ${\cal A}_1(s')\cap{\cal A}_{\rm even}\neq\emptyset$, let $a_0=\arg\min_{a\in{\cal A}_1(s')\cap{\cal A}_{\rm even}}f_a$. Obviously, we can rearrange the order of type-$1$ elements in ${\cal A}_1(s')\cap{\cal A}_{\rm even}\neq\emptyset$ such that $a_0$ is the right end of all even indexed type-$1$ treatments. Take $s_0'=(1,5,1,6,1,2,1,2,4,4,4,3,3,3,3)$ as example. It can be rearranged as $(1,2,1,2,1,5,1,6,4,4,4,3,3,3,3)$. If $f_{\rm iso}\geq f_{a_0}+2$, then we can enlarge $f_{a_{\rm iso}}$ by one and decrease $f_{a_0}$ by changing the first period assigned with treatment $a_{\rm iso}$ to $a_0$. In this process, $\chi_{s'}$ is strictly decreased while $\psi_{s'}$ and $\gamma_{s'}$ are unchanged. We can keep doing this until $f_{\rm iso}\leq f_{a_0}+1$. If $f_{\rm iso}=0$, the result is verified. If $0<f_{\rm iso}\leq f_{a_0}+1$ and $|{\cal A}_1|$ is even, we can balance the frequencies of type-1 treatments and sort all treatments of other types.  This process will not increase $\chi_s$ but will strictly increase $\gamma_s$ and thus $g_{s^{''}}(x)>g_{s^{'}}(x)$ for arbitrary $x\in[0.4,0.5)$. If $0<f_{\rm iso}\leq f_{a_0}+1$ and $|{\cal A}_1|$ is odd, we can change $a_{\rm iso}$ to a type-1 treatment, balance its frequency with all other type-1 treatments and sort all treatments of other types. Let $s^{''}_1$ denote the resulting array. It can be seen that $\chi_s$ is not increased, $\psi_s$ is increased by $1$ and $\gamma_s$ is decreased by at most $f_{\rm iso}-2$. On the other hand, change $a_0$ in to a type-2 treatment, balance the frequencies of all remaining type-1 treatments and sort all treatments of other types. This process will not increase $\chi_s$ but will increase $\gamma_s$ by $f_{a_0}-1$ and decrease $\psi_s$ by $1$. Let $s^{''}_2$ denote the resulting array. and thus $g_{s^{''}}(x)>g_{s^{'}}(x)$ for arbitrary $x\in[0.4,0.5)$.
\end{proof}

\begin{lemma}\label{lem:120801}
Suppose sequence $s$ has no type-$0$ treatment and can be separated into two parts: $s=(s_1|s_2)$ such that $s_1$ and $s_1$ contain only type-$1$ treatment and no treatment appear both in $s_1$ and $s_2$. Let $k_1$ denote the length of $s_1$ and $t_1$ denote the number of different treatments in $s_1$, and similarly $k_2,t_2$ for $s_2$.
Then, we generate four different types of dominating sequences, say $s'$, for $s$. Without loss of generality, we assume the treatments in $s_1$ are $1,\ldots,t_1$.
\begin{itemize}
\item[Case 1.] Even $k_1$, even $t_1$: let $1+f_{t_1/2}\geq f_1\geq f_2\cdots\geq f_{t_1/2-1}\geq f_{t_1/2}$ and $\sum^{t_1/2}_{j=1}f_j=k_1/2$. And $1+f_{t_1}\geq f_{t_1/2+1}\geq f_{t_1/2+2}\cdots\geq f_{t_1-1}\geq f_{t_1}$ and $\sum^{t_1}_{j=t_1/2+1}f_j=k_1/2$. 
\item[Case 2.] Even $k_1$, odd $t_1$: let $1+f_{(t_1+1)/2}\geq f_1\geq f_2\cdots\geq f_{(t_1+1)/2-1}\geq f_{(t_1+1)/2}$ and $\sum^{(t_1+1)/2}_{j=1}f_j=k_1/2$. And $1+f_{t_1}\geq f_{(t_1+1)/2+1}\geq f_{(t_1+1)/2+2}\cdots\geq f_{t_1-1}\geq f_{t_1}$ and $\sum^{t_1}_{j=(t_1+1)/2+1}f_j=k_1/2$.
\item[Case 3.] Odd $k_1$, even $t_1$: let $1+f_{t_1/2}\geq f_1\geq f_2\cdots\geq f_{t_1/2-1}\geq f_{t_1/2}$ and $\sum^{t_1/2}_{j=1}f_j=(k_1+1)/2$. And $1+f_{t_1}\geq f_{t_1/2+1}\geq f_{t_1/2+2}\cdots\geq f_{t_1-1}\geq f_{t_1}$ and $\sum^{t_1}_{j=t_1/2+1}f_j=(k_1-1)/2$.
\item[Case 4.] Odd $k_1$, odd $t_1$: let $1+f_{(t_1+1)/2}\geq f_1\geq f_2\cdots\geq f_{(t_1+1)/2-1}\geq f_{(t_1+1)/2}$ and $\sum^{(t_1+1)/2}_{j=1}f_j=(k_1+1)/2$. And $1+f_{t_1}\geq f_{(t_1+1)/2+1}\geq f_{(t_1+1)/2+2}\cdots\geq f_{t_1-1}\geq f_{t_1}$ and $\sum^{t_1}_{j=(t_1+1)/2+1}f_j=(k_1-1)/2$.

\end{itemize}
Sequentially assign all treatments $1,\ldots, t_1/2$ to the odd periods $1,3,5$ and so on. Sequentially assign all treatments $t_1/2+1,\ldots, t_1$ to the even periods $2,4,6$ and so on. The resulting sub-sequence is denoted by $s_1'$. 
Rearrange the $t_2$ treatments in $s_2$ such that they have balanced frequencies and are sorted in ascending order such that all of them become type-$2$ treatments. The resulting sub-sequence is denoted by $s_2'$. 
Let $s'=(s_1'|s_2')$.
We have $q_{s'}(0.5)\geq q_{s}(0.5)$.
Given $(k,t,k_1,t_1,t_2)$ ($k_2=k-k_1$), $s'$ can be uniquely determined and is so denoted by $s'(k,t,k_1,t_1,t_2)$. Note that the two extreme cases with zero $s_2'$ length and zero $s_2'$ length are also included, which means the rearranging process above is carried out over ($i$) $s=s_1$ to derive $s'=s_1'$ and ($ii$) $s=s_2$ to derive $s'=s_2'$, respectively.
For convenience, we call $s'(k,t,k_1,t_1,t_2)$ the {\rm winner sequence}. If a sequence $s$ has the same $\chi$, $\gamma$ and $\psi$ value as $s'(k,t,k_1,t_1,t_2)$, we call $s$ the {\it competitor} of $s'(k,t,k_1,t_1,t_2)$. Given $(k,t,k_1,t_1,t_2)$, we call the corresponding $s'(k,t,k_1,t_1,t_2)$ and all its competitors as {\it leading sequence}, which is denoted by ${\cal L}(k,t,k_1,t_1,t_2)$.
\end{lemma}
\begin{proof} The proof of this lemma includes sophisticated but strait-forward analysis. The detailed proof is kept in some unreported works for better reading experience. We provide an intuitive proof of this lemma as follows, which shall help readers to address this issue. Note that all treatments in $s_1$ are type-$1$. We can tear $s_1$ into two parts: odd periods and even periods. If any treatments appear both in an odd period and an even period, it can not be of type-$1$. Thus, each treatment in $s_1$ must appear either in odd periods or even periods together. The number of even-index periods equals to or is one less than that of odd-index periods. Suppose there are $t_1^*$ and $t_2^*$ treatments in even and odd periods respectively. We can rebalance the frequencies of treatments in even periods and do the same for the odd periods to decrease the $\chi$ value. Note that no other parameter in $q_{s}(0.5)$ is changed in this process other than the $\chi$ value. We know, $q_{s}(0.5)$ is increased. When $|t_2^*-t_1^*|>1$, say $t_1^*-t_2^*=2$, we can move one treatment from the $t_1^*$ treatments in even periods to odd periods and then do the rebalancing again. In this process, the $\chi$ value is again decreased and so $q_{s}(0.5)$ increases. Thus, $|t_1^*-t_2^*|\leq 1$. When $t^*_1\neq t_2^*$, we put more treatments in odd periods since it is longer. By doing so, $\chi$ is decreased as many as possible and so $q_{s}(0.5)$ is increased. We call $s''$ a {\it competitor} of $s'$ if it has the same $\chi$, $\psi$ and $\gamma$ value as $s'$. It is obvious that $q_{s*}(0.5)=\max_{s\in{\cal S}} q_s(0.5)$ indicates $s^*$ is $s'$ or a competitor of $s'$.
\end{proof}

{\noindent\bf Proof of Theorem \ref{lem:121001}.} From Lemma \ref{lem:b1}, we know that there is no sequence in ${\cal T}$ which has type-$0$ treatments. For any $s\in{\cal T}$, let ${\cal A}_1(s)$ denote all type-$1$ treatments in $s$. Now we change the order of all these treatments while keeping their frequencies unchanged as follows.

 For an arbitrary $s\in{\cal T}$, let $n_{i}$ denote the number of type-$i$ treatments. According to Lemma \ref{lem:b1}, we know $n_{0}=0$. Suppose there are $n_1$ type-$1$ treatments. If we look at one type-$1$ treatment, say treatment $1$, alone, it forms a subsequence appearing in positions $i_1,i_1+2,i_1+4,\ldots,i_1+{2f_1-2}$ for some $1\leq i_1$ and $i_1+{2f_1-2}\leq k$.

Consider now finding an $s'\in{\cal S}^*$ such that $s'$ {\it dominates} $s$, which means $q_{s'}(x)>q_{s}(x)$ for all $x\in[0.4,0.5]$. If such an $s'\in{\cal S}^*$ exists for each $s\not\in{\cal S}^*$, the result of this theorem right follows.
We change the order of treatments in $s$ as follows. Let $1,\ldots,n_1$ denote the $n_1$ type-$1$ treatments. One can always find
\begin{eqnarray*}
({i_1},\ldots,i_{n_1},t')=\arg\min\bigg\{\sum_{j=1}^{t'}f_{i_j}-\sum_{j=t'+1}^{n_1}f_{i_j}:~\sum_{j=1}^{t'}f_{i_j}-\sum_{j=t'+1}^{n_1}f_{i_j}\geq 0, ~0\leq t'\leq n_1,\\
(i_1,\ldots,i_{n_1}) {\rm~is ~a ~permutation ~of} ~1,\ldots,n_1\bigg\}.
\end{eqnarray*}
Beginning from the first period, we arrange the $\sum_{j=1}^{t'}n_{i_j}$ type-$1$ treatment in ascending order and a uniform distance $2$ which means positions $1,3,5$ and so on. Beginning from the second period, we arrange the rest $\sum_{j=t'+1}^{n_1}n_{i_j}$ type-$1$ treatment in ascending order and a uniform distance $2$ which means positions $2,4,6$ and so on. On the empty periods which have not been assigned any treatment yet, we arrange the rest according to their frequencies. The treatments with smaller frequencies are arranged at left periods. The resulting array is denoted by $s_{\rm temp}$. To make it clear, an illustrative example is given in the following paragraph.

Here we shall give a toy example to show how this is done. Suppose $s=(1,2,1,2,1,3,4,3,$ $5,3,6,6,6,5,5,4,4,4)$. Then, there are three type-$1$ treatments, i.e., $1,2$ and $3$. The minimum non-negative value of $\sum_{j=1}^{t'}n_{i_j}-\sum_{j=t'+1}^{n_1}n_{i_j}$ is $f_1+f_2-f_3=3+2-3=2$. We assign treatment $1$ and $2$ in ascending order to periods $1,3,5$ and so on as $(1,\#,1,\#,1,\#,2,\#,2,\#,$ $\#,\ldots)$ where $\#$ represents undetermined periods. And then, we assign $f_3=3$ treat $3$ to periods $2,4,6$ and the resulting array is $(1,3,1,3,1,3,2,\#,2,\#,\#,$ $\ldots)$. For the rest three treatments, i.e., $5,6$ and $4$, their frequencies are $f_4=4$, $f_5=3$ and $f_6=3$, respectively. So, treatments $4,6$ and $5$ are arranged from the left to the right as $s_{\rm temp}=(1,3,1,3,1,3,2,5,2,5,5,6,6,6,4,4,4,4)$. For the resulting array, we know that all type-$1$ treatments are still type-$1$ and other treatments are all type-$v$ ($v\leq 2$, we can not exclude new type-$1$ treatments which will be shown in next paragraph).

It is possible that, in this process, the number of type-$1$ treatments is increased. For example, the original sequence is $s=(1,2,1,2,1,3,4,3,5,3,6,6,6,5,4,4,4)$, i.e., the frequency of $5$ is decreased by one. Now the new resulting array is $s_{\rm temp}=(1,3,1,3,1,3,2,5,2,$ $5,6,6,6,4,4,4,4)$ and there is one more type-$1$ treatment $5$. We choose to repeat this process until there is no more new type-$i$ treatments. And the resulting array is, with a mild abuse of notations, denoted by $s_{\rm temp}$. It should be mentioned that, the resulting array has maximum $\gamma$ value, maximum $\psi$ value and minimum $\chi$ value. Meanwhile, in this process, $q_{s}(0.5)$ increases as the number pf type-$1$ treatments increases, i.e., $q_{s_{\rm temp}}(0.5)\geq q_s(0.5)$.

If $\sum_{j=1}^{t'}n_{i_j}-\sum_{j=t'+1}^{n_1}n_{i_j}\leq 1$, we know that the first $k_1=\sum_{j=1}^{t'}f_{i_j}+\sum_{j=t'+1}^{f_1}n_{i_j}$ periods are filled by these type-$1$ treatments. For the rest periods, we can sort the other treatments in the ascending order. Note that these treatments are type-$v$, $v\geq 2$. We know that this sorted version has all other treatments as type-$2$ treatment. In this process, $\gamma$ is increased or unchanged. And then, we rebalance the frequencies of these type-$2$ treatments such that $\chi$ is decreased or unchanged. So its $\psi$ value is equal to or greater than that of $s$. The sequence can be rearranged as the sequence $s'$ in Lemma \ref{lem:120801}.

\vspace{.4cm}

{\noindent\bf Proof of Theorem \ref{thm:121101}.} Theorem \ref{thm:121101} can be proved directly by Theorem \ref{lem:121001}.

\end{document}